\pdfoutput=1

\documentclass[letterpaper,twocolumn,10pt,journal]{IEEEtran}

\RequirePackage{fix-cm}

\usepackage{balance}

\usepackage{amssymb,amsmath,amsthm}
\usepackage{multirow}%
\usepackage{enumerate}%
\usepackage{comment}%
\usepackage{url}%
\usepackage{graphicx}%
\graphicspath{{./figures/}}%
\usepackage{paralist}%
\usepackage[active]{srcltx}%
\usepackage{color}%
\usepackage{algorithm}%
\usepackage{algpseudocode}%
\usepackage{tikz}
\usetikzlibrary{positioning, shapes.misc, fit}
\usepackage{caption}
\usepackage{subcaption}

\usepackage{hyperref}
 \newtheorem{theorem}{Theorem}
 
 \newtheorem{lemma}{Lemma}
 \newtheorem{definition}{Definition}

\newtheorem{rem}{Remark}

\algrenewcommand\algorithmicrequire{\textbf{\quad Input:}}
\algrenewcommand\algorithmicensure{\textbf{\quad Output:}}
\newenvironment{proof sketch}[1]{\noindent {\emph{Proof sketch of #1:}}}{\hfill \qed}

\newcommand{\eps}{\varepsilon}

\newcommand{\know}{\text{\textit{know}}}
\newcommand{\true}{\text{\textbf{true}}}
\newcommand{\false}{\text{\textbf{false}}}

\newcommand{\byz}[1]{\textsf{Byz}$\left(#1\right)$}
\newcommand{\om}[1]{\textsf{Om}$\left(#1\right)$}

\newcommand{\NN}{{\mathbb{N}}}

\newcommand{\Prr}[2][]{\Pr_{#1}{\left[#2\right]}}

\usepackage{color}
\usepackage{soul}

\begin{document}

\title{Robust Routing Made Easy:\\
Reinforcing Networks Against Non-Benign Faults
\thanks{
Research supported by the Federal Ministry of Education and Research (BMBF), grant 16KISK020K, 2021-2025. 
This article extends work presented at SSS
2017~\cite{DBLP:conf/sss/LenzenM17}.}}

%\date{}

\author{Christoph Lenzen$^1$ \quad Moti Medina$^2$ \quad Mehrdad Saberi$^3$ \quad Stefan Schmid$^4$\\
{\small
$^1$CISPA Helmholtz Center for Information Security, Germany \quad $^2$Faculty of Engineering, Bar-Ilan University, Ramat Gan, Israel}\\
{\small$^3$University of Maryland, College Park, USA \quad $^4$TU Berlin, Germany}
}

%\author{~~~~~~~~~~~~~~~~~~~~~~Christoph Lenzen \hfill Max Planck Institute for
%Informatics, Germany~~~~~~~~~~~~~~~~~~~~~~\\
%~~~~~~~~~~~~~~~~~~~~~~Moti Medina \hfill Ben-Gurion University of the Negev,
%Israel~~~~~~~~~~~~~~~~~~~~~~\\
%~~~~~~~~~~~~~~~~~~~~~~Mehrdad Saberi\hfill Sharif University of Technology,
%Iran~~~~~~~~~~~~~~~~~~~~~~\\
%~~~~~~~~~~~~~~~~~~~~~~Stefan Schmid\hfill University of Vienna, Austria~~~~~~~~~~~~~~~~~~~~~~}


\maketitle

\begin{abstract}
With the increasing scale of communication networks,
the likelihood of failures grows as well.
Since these networks form a critical backbone
of our digital society, it is important that they rely on
robust routing algorithms which ensure connectivity
despite such failures. While most modern communication
networks feature robust routing mechanisms, these mechanisms
are often fairly complex to design and verify, as they
need to account for the effects of failures and rerouting
on communication.

This paper conceptualizes the design of robust routing mechanisms,
with the aim to avoid such complexity. In particular,
we showcase \emph{simple} and generic blackbox transformations that increase resilience of routing against independently distributed failures, which allows
to simulate the routing scheme on the original network, even in the presence
of non-benign node failures (henceforth called faults). This is attractive
as the system specification and routing policy can simply be preserved.

We present a scheme for constructing such a reinforced network, given
an existing (synchronous) network and a routing scheme. We prove that
this algorithm comes with small constant overheads, and only requires a minimal
amount of additional node and edge resources; 
in fact, if the failure probability is smaller than $1/n$, 
the algorithm can come without any overhead at all. 
At the same time,
it allows to tolerate a large number of
independent random (node) faults,
asymptotically almost surely.
We complement our analytical results with simulations on different real-world topologies.
\end{abstract}

%\keywords{Fault-tolerance \and Reinforcement \and Grids \and Minor-free Graphs}

\section{Introduction}

Communication networks have become a critical backbone
of our digital society. For example, many datacentric applications
related to entertainment, social networking, or health, among others,
are distributed and rely on the high availability and
dependability of the interconnecting network (e.g., a
datacenter network or a wide-area network).
At the same time, with the increasing scale of
today's distributed and networked systems (often relying
on commodity hardware as a design choice
\cite{al2008scalable,barroso2003web,ghemawat2003google}), the number of
failures is likely to increase as well
\cite{gill2011understanding,failures-uninett,link-failures-ip-backbone,single-failure,single-failure-journal}.
It is hence important that communication networks can tolerate
such failures and
remain operational despite the failure of some of their
components.

Robust routing mechanisms aim to provide such guarantees:
by rerouting traffic quickly upon failures,
reachability is preserved. Most communication
networks readily feature robust routing mechanisms,
in the control plane (e.g.
\cite{francois2005achieving,gafni-lr,greenberg2005clean,oran1990rfc1142}), in
the data plane (e.g. \cite{ipfrr,conext18,ddc,mplsfrr}), as well as on higher
layers (e.g. \cite{andersen2001resilient}).
However, the design of such robust routing mechanisms is
still challenging and comes with tradeoffs, especially if
resilience should extend to multiple failures \cite{chiesa2016resiliency}.

Besides a fast reaction time and re-establishing connectivity, the
resulting routes typically need to fulfill certain additional properties,
related to the network specification and policy.
Ensuring such properties however can be fairly complex,
as packets inevitably follow different paths after failures.
Interestingly, while the problem of how to re-establish reachability
after failures is well explored,
the problem of providing specific properties on the failover
paths is much less understood.

This paper conceptualizes the design of robust routing, presenting a new approach to robust routing which conceptually differs
significantly from existing literature by relying on \emph{proactive reinforcement} (rather than reaction to failures).
In particular, our approach aims to overcome the complexities involved in designing
robust routing algorithms, by simply sticking to the original
network and routing specification.
To achieve this, our approach is to mask the effects of failures
using \emph{redundancy}: in the spirit of error correction,
we proactively reinforce networks by adding a minimal number of
additional nodes and links, rather than
coping with failed components when they occur.
The latter is crucial
for practicability: significant refactoring of existing systems
and/or accommodating substantial design constraints is rarely
affordable.

In this paper, to ensure robustness while maintaining
the network and routing specification, we aim to
provide a high degree of fault-tolerance,
which goes beyond simple equipment and failstop failures,
but accounts for more general \emph{faults} which include non-benign
failures of entire nodes.

While our approach presented in this paper will be general
and applies to \emph{any} network topology, we are particularly
interested in datacenter networks (e.g., based on low-dimensional
hypercubes or $d$-dimensional tori \cite{guo2009bcube,leighton2014introduction})
as well as in wide-area
networks (which are typically sparse \cite{spring2002measuring}).
We will show that our approach works especially well for these networks.

%###
\subsection{The Challenge}
%###
% The task we set out for ourselves is the following.

More specifically,
% our paper revolves around
% the following challenge.
we are given a network $G=(V,E)$ and a routing scheme, i.e.,
a set of routes in $G$.
We seek to reinforce the network $G$ by
allocating additional resources, in terms of nodes and edges,
and to provide a corresponding routing strategy to simulate the routing scheme
on the original network despite non-benign node failures.

The main goal is to maximize the probability that the network withstands
failures (in particular, random failures of entire nodes),
while minimizing the resource overhead.
 Furthermore, we want to ensure that the network transformation is simple
to implement, and that it interferes as little as possible with the existing system design and operation, e.g., it
does not change the reinforced system's specification.

Toward this goal, in this paper, we make a number of simplifying assumptions.
First and most notably, we assume independent failures,
that is, we aim at masking faults with little or no correlation among each other.
Theoretically, this is motivated by the fact that
%highly correlated faults
%necessitate high degrees of redundancy and thus overheads;
%for instance, a system-wide power outage, whether rare or not, cannot be addressed by %adding extra nodes or edges that are connected to the same power source,
%but requires independent backup power.
%More generally,
guaranteeing full functionality despite having $f$ adversarially placed faults trivially requires redundancy (e.g., node degrees) larger than $f$.
There is also practical motivation to consider independent faults:
many distributed systems proactively avoid fault clusters
\cite{feldmann2016netco,scheideler2005spread} and there is also empirical
evidence that in certain scenarios, failures are only weakly correlated \cite{lee2010diverse}.
%As there are many reasons why topologies of communication networks feature
%very small degrees in practice, assuming worst-case \emph{distribution}
%of faults would hence come at too high of a cost.

Second, we treat nodes and their outgoing links as fault-containment regions (according to \cite{kopetz03}), i.e., they are the basic components our systems are comprised of.
This choice is made for the sake of concreteness;
similar results could be obtained when considering, e.g., edge failures, without changing the gist of results or techniques.
With these considerations in mind, the probability of uniformly random
node failures that the reinforced system can tolerate is a canonical choice for measuring resilience.

Third, we focus on synchronous networks, for
several reasons:
synchrony not only helps in handling faults, both on the theoretical level (as illustrated by the famous FLP theorem~\cite{fischer85impossibility}) and for ensuring correct implementation, but it also
simplifies presentation, making it easier to focus on the proposed concepts.
In this sense, we believe
that our approach is of particular interest in the context of real-time systems,
where the requirement of meeting hard deadlines makes synchrony an especially attractive choice.

%###
\subsection{Contributions and Techniques}
%###

This paper proposes a novel and simple approach to robust routing,
which decouples the task of designing a reinforced network from the task of
designing a routing scheme over the input network. By virtue of this decoupling,
our approach supports arbitrary routing schemes and objectives,
from load minimization to throughput maximization and beyond,
in various models of computation, e.g., centralized or distributed, randomized
or deterministic, online or offline, or oblivious.

We first consider a trivial approach:
we simply replace each node by $\ell \in \NN$ copies
and for each edge we connect each pair of copies of its endpoints,
where $\ell$ is a constant.\footnote{Choosing concreteness over generality,
we focus on the, in our view, most interesting case of constant $\ell$. It is straightforward to generalize the analysis.}
Whenever a message would be sent over an edge in the original graph,
it should be sent over each copy of the edge in the reinforced graph.
If not too many copies of a given node fail, this enables each receiving copy to recover the correct message.
Thus, each non-faulty copy of a node can run the routing algorithm as if it were the original node, guaranteeing that it has the same view of the system state as its original in the corresponding fault-free execution of the routing scheme on the original graph.

When analyzing this approach,
we observe that asymptotically almost surely (a.a.s., with probability $1-o(1)$) and with $\ell=2f+1$, this reinforcement can sustain an independent probability $p$ of $f$ \emph{Byzantine} node failures~\cite{pease80}, for any $p\in o(n^{-1/(f+1)})$, i.e., $f$  nodes may violate the protocol in any arbitrary way (and may hence also collude).
This threshold is sharp up to (small) constant factors: for $p\in \omega(n^{-1/(f+1)})$, a.a.s.\ there is some node for which all of its copies  fail.
If we restrict the fault model to omission faults
(faulty nodes may skip sending some messages but otherwise act according to the protocol), $\ell=f+1$ suffices.
The cost of this reinforcement is that the number of nodes and edges increase by factors of $\ell$ and $\ell^2$, respectively.
Therefore, already this simplistic solution can support non-crash faults of probability $p\in o(1/\sqrt{n})$ at a factor-$4$ overhead.

We note that the simulation introduces no large computational overhead and
does not change the way the system works, enabling to use it as a blackbox.
Also randomized algorithms can be simulated in a similar fashion,
provided that all copies of a node have access to a shared source of randomness.
Note that this requirement is much weaker than globally shared randomness:
it makes sense to place the copies of a node in physical proximity to approximately preserve the geometrical layout of the physical realization of the network topology.

Our approach above raises the question whether
we can reduce the involved overhead further.
In this paper, we will answer this question positively:
We propose to apply the above strategy only to a small
subset $E'$ of the edge set.
Denoting by $v_1,\ldots,v_{\ell}$ the copies of node $v\in V$, for
any remaining edge $\{v,w\}\in E\setminus E'$ we add only edges
$\{v_i,w_i\}$, $i\in [\ell]$, to the reinforced graph.
The idea is to choose $E'$ in a way such that the connected components
induced by $E\setminus E'$ are of constant size, yet $|E'|=\varepsilon |E|$.
This results in the same asymptotic threshold for $p$, while the number of edges of the reinforced graph drops to $((1-\varepsilon)\ell+\varepsilon \ell^2)|E|$.
For any constant choice of $\varepsilon$, we give constructions with this property for grids or tori of constant dimension and minor-free graphs of bounded degree.
Again, we consider the case of $f=1$ of particular interest:
in many typical network topologies, we can reinforce the network to boost the failure probability that can be tolerated from $\Theta(1/n)$ to $\Omega(1/\sqrt{n})$ by roughly doubling (omission faults) or tripling (Byzantine faults) the number of nodes and edges.

The redundancy in this second construction is near-optimal under the constraint that we want to simulate an arbitrary routing scheme in a blackbox fashion,
as it entails that we need a surviving copy of each edge, and thus in particular each node.
In many cases, the paid price will be smaller than the price for making each individual component sufficiently reliable to avoid this overhead.
Furthermore, we will argue that the simplicity of our constructions enables us to re-purpose the redundant resources in applications with less strict reliability requirements.

Our results show that while approach is general and can be applied to any
existing network topology (we will describe and analyze valid reinforcements for
our faults models on general graphs), it can be refined and is particularly
interesting in the context of networks that
admit suitable partitionings. Such networks include
sparse, minor-free graphs, which are practically relevant topologies in
wide-area networks, as well as torus graphs and low-dimensional
hypercubes, which arise in datacenters and parallel architectures.

To complement our theoretical findings and investigate the reinforcement
cost in real networks, we conducted experiments on the Internet Topology Zoo~\cite{knight2011internet}.
We find that our approach achieves robustness at significantly lower cost compared to
the naive replication strategy often employed in dependable networks.

%###
\subsection{Putting Things Into Perspective}
%###

In contrast to much existing robust routing literature on  \emph{reactive}
approaches to link failures~\cite{frr-survey} (which come with a delay),
we consider a proactive approach by enhancing the network with redundancy.
Our proactive approach also allows us to replicate the routing scheme (and hence the network policy) on the new network.
In particular, we show that if the failure probability is smaller than $1/n$, there is a good probability that our approach works even without any overhead at all. 
Furthermore, there are two ways in which our system can be used. One approach is to replicate the entire node (including the compute part), and then forward the traffic to its two associated peers. Alternatively, traffic can also simply be replicated to multiple NICs, without additional compute requirements, depending on the failure model. More generally, our contribution can also be seen more abstractly and the robust routing happen on a logical level, depending on the failure scenario. 
Also, we show that in the absence of a valid message, it can simply be ignored, as the rest of the system continues to perform

The most closely related work to ours is NetCo~\cite{disn16netco},
which also relies on network reinforcement and can handle malicious behavior.
NetCo is is based on a robust
combiner concept known from cryptography, and complements each router with two additional routers.
Using software-defined networking,  traffic is replicated across the three (untrusted) devices and then merged again, using a consensus algorithm. While a high degree of robustness is achieved, the three-fold overhead is significant. More importantly, however, in contrast to our approach, Netco requires special hardware for splitting and merging the traffic; while the functionality of this hardware can be simple, it still needs to be trusted. The consensus requirement dramatically reduces the throughput, as shown in the empirical evaluation of NetCo in~\cite{disn16netco}.

Our solution does not require such components and is hence not only more practical but also significantly more performant.

\subsection{Organization}

In \S~\ref{sec:applications}, we sketch the properties of our approach and state a number of potential applications. In \S~\ref{sec:prelim}, we formalize the fault models that we tackle in this article alongside the notion of a valid reinforcement and its complexity measures. In \S~\ref{sec:strongbyz} and \S~\ref{sec:strong_om}, we study valid reinforcements on general graphs, and in \S~\ref{sec:eff}, we study more efficient reinforcements for specific graphs.
We complement our analytical results with an empirical simulation study in
\S~\ref{sec:eval}.
In \S~\ref{sec:disc} we raise a number of points in favor of the reinforcement approach. We review related work in
\S~\ref{sec:relwork}, and we conclude and present a number of interesting
follow-up questions in \S~\ref{sec:conc}.

\section{High-level Overview: Reinforcing Networks}\label{sec:applications}

Let us first give an informal overview of our blackbox transformation
for reinforcing networks (for formal specification see \S~\ref{sec:prelim}), as well as its guarantees and preconditions.

%###
\paragraph{Assumptions on the Input Network}
%###
We have two main assumptions on the network at hand: (1)~We consider synchronous routing networks, and (2)~each node in the network (alongside its outgoing links) is a fault-containment region, i.e., it fails independently from other nodes.
We do not make any assumptions on the network topology, but will provide specific
optimizations for practically relevant topologies (such as sparse, minor-free networks
or hypercubes) in \S~\ref{sec:eff}.

%###
\paragraph{Valid Reinforcement Simulation Guarantees}
%###
Our reinforcements create a number of copies of each node. We have each non-faulty copy of a node run the routing algorithm as if it were the original node, guaranteeing that it has the same view of the system state as its original in the corresponding fault-free execution of the routing scheme on the original graph. Moreover, the simulation fully preserves all guarantees of the schedule, including its timing, and introduces no big computational overhead.
This assumption is simple to meet in stateless networks, while it requires synchronization primitives in case of stateful network functions.

%###
\paragraph{Unaffected Complexity and Cost Measures}
%###
Routing schemes usually revolve around objective functions such as load minimization, maximizing the throughput, minimizing the latency, etc., while aiming to minimize complexity related to, e.g., the running time for centralized algorithms, the number of rounds for distributed algorithms, the message size, etc. Moreover, there is the degree of uncertainty that can be sustained, e.g., whether the input to the algorithm is fully available at the beginning of the computation (offline computation) or revealed over time (online computation). Our reinforcements preserve all of these properties, as they operate in a blackbox fashion. For example, our machinery readily yields various fault-tolerant packet routing algorithms in the Synchronous Store-and-Forward model by Aiello et.\ al~\cite{AKOR}. More specifically, from~\cite{spaaEvenMP15} we obtain a centralized deterministic online algorithm on unidirectional grids of constant dimension that achieves a competitive ratio which is polylogarithmic in the number of nodes of the input network w.r.t.\ throughput maximization. Using~\cite{EvenMR16} instead, we get a centralized randomized offline algorithm on the unidirectional line with constant approximation ratio w.r.t.\ throughput maximization. In the case that deadlines need to be met the approximation ratio is, roughly, $O(\log^* n)$~\cite{RackeR11}. As a final example, one can obtain from~\cite{AKK} various online distributed algorithms with sublinear competitive ratios w.r.t.\ throughput maximization.

%###
\paragraph{Cost and Gains of the Reinforcement}
%###
The price of adding fault-tolerance is given by the increase in the network size, i.e., the number of nodes and edges of the reinforced network in comparison to the original one. Due to the assumed independence of node failures, it is straightforward to see that the (uniform) probability of sustainable node faults increases roughly like $n^{-1/(f+1)}$ in return for (i) a linear-in-$f$ increase in the number of nodes and (ii) an increase in the number of edges that is quadratic in $f$. We then proceed to improve the construction for grids and minor-free constant-degree graphs to reduce the increase in the number of edges to being roughly linear in $f$. Based on this information, one can then assess the effort in terms of these additional resources that is beneficial, as less reliable nodes in turn are cheaper to build, maintain, and operate. We also note that, due to the ability of the reinforced network to ensure ongoing unrestricted operability in the presence of some faulty nodes, faulty nodes can be replaced or repaired \emph{before} communication is impaired or breaks down.

%###
\paragraph{Preprocessing}
%###
Preprocessing is used, e.g., in computing routing tables in Oblivious Routing~\cite{racke2009survey,esa19}.
The reinforcement simply uses the output of such a preprocessing stage in the same manner as the original algorithm. In other words, the preprocessing is done on the input network and its output determines the input routing scheme. In particular, the preprocessing may be randomized and does not need to be modified in any way.

%###
\paragraph{Randomization}
%###
Randomized routing algorithms can be simulated as well, provided that all copies of a node have access to a shared source of randomness. We remark that, as our scheme locally duplicates the network topology, it is natural to preserve the physical realization of the network topology in the sense that all (non-faulty) copies of a node are placed in physical proximity. This implies that this constraint is much easier to satisfy than globally shared randomness.

\section{Preliminaries}\label{sec:prelim}
% \subsection{The Network}
We consider synchronous routing networks.
Formally, the network is modeled as a directed graph $G=(V,E)$, where $V$ is the set of $n\triangleq |V|$ vertices, and $E$ is the set of $m\triangleq |E|$ edges (or links).
Each node maintains a state, based on which it decides in each round for each of its outgoing links which message to transmit.
We are not concerned with the inner workings of the node, i.e., how the state is updated;
rather, we assume that we are given a scheduling algorithm performing the task of updating this state and use it in our blackbox transformations.
In particular, we allow for online, distributed, and randomized algorithms.

%###
\paragraph{Probability-$p$ Byzantine Faults \byz{p}}
%###
The set of faulty nodes $F\subseteq V$ is determined by sampling each $v\in V$ into $F$ with independent probability $p$. Nodes in $F$ may deviate from the protocol in arbitrary ways, including delaying, dropping, or forging messages, etc.
%###
\paragraph{Probability-$p$ Omission Faults \om{p}}
%###
The set of faulty nodes $F\subseteq V$ is determined by sampling each $v\in V$ into $F$ with independent probability $p$. Nodes in $F$ may deviate from the protocol by not sending a message over an outgoing link when they should. We note that it is sufficient for this fault model to be satisfied \emph{logically.} That is, as long as a correct node can identify incorrect messages, it may simply drop them, resulting in the same behavior of the system at all correct nodes as if the message was never sent.
% We stress that our model assumption of having a synchronous system is very useful here, as any timing violation will be detected. We remark that, as our goal is not to protect networks from malicious attacks, but from random failures, in practice this will not require involved cryptographic schemes; simple measures that make it sufficiently unlikely for a fault to emulate correctly generated messages is sufficient.

%###
\paragraph{Simulations and Reinforcement}
%###
For a given network $G=(V,E)$ and a scheduling algorithm $A$, we will seek to \emph{reinforce} $(G,A)$ by constructing $G'=(V',E')$ and scheduling algorithm $A'$ such that the original algorithm $A$ is \emph{simulated} by $A'$ on $G'$, where $G'$ is subject to random node failures. We now formalize these notions. First, we require that there is a surjective mapping $P:V'\to V$; fix $G'$ and $P$, and choose $F'\subseteq V'$ randomly as specified above.
\begin{definition}[Simulation under \byz{p}]
Assume that in each round $r\in \NN$, each $v'\in V'\setminus F'$ is given the same input by the environment as $P(v')$. $A'$ is a \emph{simulation} of $A$ under \byz{p}, if for each $v\in V$, a strict majority of the nodes $v'\in V'$ with $P(v')=v$ computes in each round $r\in \NN$ the state of $v$ in $A$ in this round. The simulation is \emph{strong}, if not only for each $v\in V$ there is a strict majority doing so, but all $v'\in V'\setminus F'$ compute the state of $P(v')$ in each round.
\end{definition}
\begin{definition}[Simulation under \om{p}]
Assume that in each round $r\in \NN$, each $v'\in V'$ is given the same input by the environment as $P(v')$. $A'$ is a \emph{simulation} of $A$ under \om{p}, if for each $v\in V$, there is $v'\in V'$ with $P(v')=v$ that computes in each round $r\in \NN$ the state of $v$ in $A$ in this round. The simulation is \emph{strong}, if each $v'\in V'$ computes the state of $P(v')$ in each round.
\end{definition}
\begin{definition}[Reinforcement]
A \emph{(strong) reinforcement} of a graph $G=(V,E)$ is a graph $G'=(V',E')$, a surjective mapping $P\colon V'\to V$, and a way of determining a scheduling algorithm $A'$ for $G'$ out of scheduling algorithm $A$ for $G$. The reinforcement is \emph{valid} under the given fault model (\byz{p} or \om{p}) if $A'$ is a (strong) simulation of $A$ a.a.s.
\end{definition}

%###
\paragraph*{Resources and Performance Measures.}
%###
We use the following performance measures.
\begin{enumerate}[(i)]
\item The probability $p$ of independent node failures that can be sustained a.a.s.
\item The ratio $\nu\triangleq |V'|/|V|$, i.e., the relative increase in the number of nodes.
\item The ratio $\eta \triangleq|E'|/|E|$, i.e., the relative increase in the number of edges.
\end{enumerate}
We now briefly discuss, from a practical point of view, why we do not explicitly consider further metrics that are of interest.

\subsection*{Other Performance Measures}
\begin{compactitem}
\item \emph{Latency:}
As our reinforcements require (time-preserving) simulation relations, in terms of rounds, there is no increase in latency whatsoever.
However, we note that (i) we require all copies of a node to have access to the input (i.e., routing requests) of the simulated node and (ii) our simulations require to map received messages in $G'$ to received messages of the simulated node in $G$.
Regarding (i), recall that it is beneficial to place all copies of a node in physical vicinity, implying that the induced additional latency is small.
Moreover, our constructions naturally lend themselves to support redundancy in computations as well, by having each copy of a node perform the tasks of its original;
in this case, (i) comes for free.
Concerning (ii), we remark that the respective operations are extremely simple;
implementing them directly in hardware is straightforward and will have limited impact on latency in most systems.
\item \emph{Bandwidth/link capacities.}
We consider the uniform setting in this work.
Taking into account how our simulations operate, one may use the ratio $\eta$ as a proxy for this value.
\item \emph{Energy consumption.}
Regarding the energy consumption of links, the same applies as for bandwidth.
The energy nodes use for routing computations is the same as in the original system, except for the overhead induced by Point (ii) we discussed for latency.
Neglecting the latter, the energy overhead is in the range $[\min\{\nu,\eta\},\max\{\nu,\eta\}]$.
\item \emph{Hardware cost.}
Again, neglecting the computational overhead of the simulation, the relative overhead lies in the range $[\min\{\nu,\eta\},\max\{\nu,\eta\}]$
%\mnote{Don't get this one. Increasing the number of ports (maybe by adding an interface like a buffer-tree) is captured by this range? Maybe we should just say that eta and nu capture this o.head?} \cnote{Do they? At least not as a relative factor, as this depends on the amount of computations performed by the nodes in the input scheme.}
\end{compactitem}
In light of these considerations, we focus on $p$, $\nu$, and $\eta$ as key metrics for evaluating the performance of our reinforcement strategies.

\section{Strong Reinforcement under \texorpdfstring{\byz{p}}{Byz(p)}}\label{sec:strong_byz}\label{sec:strongbyz}

We now present and analyze valid reinforcements
under \texorpdfstring{\byz{p}}{Byz(p)}
for our faults model
on general graphs.
Given are the input network $G=(V,E)$ and scheduling algorithm $A$. Fix a parameter $f\in \NN$ and set $\ell = 2f+1$.

%###
\paragraph{Reinforced Network $G'$}
%###
We set $V'\triangleq V\times [\ell]$, where $[\ell]\triangleq \{1,\ldots,\ell\}$, and denote $v_i\triangleq (v,i)$. Accordingly, $P(v_i)\triangleq v$. We define $E'\triangleq \{(v',w')\in V'\times V'\,|\,(P(v'),P(w'))\in E\}$.

%###
\paragraph{Strong Simulation $A'$ of $A$}
%###
Consider node $v'\in V'\setminus F'$. We want to maintain the invariant that in each round, each such node has a copy of the state of $v=P(v')$ in $A$. To this end, $v'$
\begin{compactenum}[\bfseries (1)]
\item initializes local copies of all state variables of $v$ as in $A$,
\item \sloppy sends on each link $(v',w')\in E'$ in each round the message $v$ would send on $(P(v'),P(w'))$ when executing $A$, and
\item for each neighbor $w$ of $P(v')$ and each round $r$, updates the local copy of the state of $A$ as if $v$ received the message that has been sent to $v'$ by at least $f+1$ of the nodes $w'$ with $P(w')=w$ (each one using edge $(w',v')$).
\end{compactenum}
Naturally, the last step requires such a majority to exist; otherwise, the simulation fails. We show that $A'$ can be executed and simulates $A$ provided that for each $v\in V$, no more than $f$ of its copies are in $F'$.
\begin{lemma}\label{lemma:sim_byz}
If for each $v\in V$, $|\{v_i\in F'\}|\leq f$, then $A'$ strongly simulates $A$.
\end{lemma}
\begin{proof}
We show the claim by induction on the round number $r\in \NN$, where we consider the initialization to anchor the induction at $r=0$. For the step from $r$ to $r+1$, observe that because all $v'\in V'\setminus F'$ have a copy of the state of $P(v')$ at the end of round $r$ by the induction hypothesis, each of them can correctly determine the message $P(v')$ would send over link $(v,w)\in E$ in round $r+1$ and send it over each $(v',w')\in E$ with $P(w')=w$. Accordingly, each $v'\in V'\setminus F'$ receives %with $(P(w'),P(v'))\in E$
 the message $A$ would send over $(w,v) \in E$  %this link
 from each $w'\in V'\setminus F'$ with $P(w')=w$ (via the link $(w',v')$). By the assumption of the lemma, we have at least $\ell-f=f+1$ such nodes, implying that $v'$ updates the local copy of the state of $A$ as if it received the same messages as when executing $A$ in round $r+1$. Thus, the induction step succeeds and the proof is complete.
\end{proof}

%###
\paragraph{Resilience of the Reinforcement}
%###
We now examine how large the probability $p$ can be for the precondition of Lemma~\ref{lemma:sim_byz} to be satisfied a.a.s.
\begin{theorem}\label{thm:strong_byz}
If $p \in o(n^{-1/(f+1)})$, the above construction is a valid strong reinforcement for the fault model \byz{p}. If $G$ contains $\Omega(n)$ nodes with non-zero outdegree, $p\in \omega(n^{-1/(f+1)})$ implies that the reinforcement is not valid.
\end{theorem}
\begin{proof}
By Lemma~\ref{lemma:sim_byz}, $A'$ strongly simulates $A$ if for each $v\in V$, $|\{v_i\in F'\}|\leq f$. If $p \in o(n^{-1/(f+1)})$, using $\ell=2f+1$ and a union bound we see that the probability of this event is at least
\begin{align*}
&1-n\sum_{j=f+1}^{2f+1}\binom{2f+1}{j}p^j(1-p)^{2f+1-j}\\
\geq\,& 1-n \sum_{j=f+1}^{2f+1}\binom{2f+1}{j}p^{j}\\
\geq\,& 1-n \binom{2f+1}{f+1}p^{f+1}\sum_{j=0}^f p^j\\
\geq\,& 1-n (2e)^f\cdot\frac{p^{f+1}}{1-p}= 1-o(1).
\end{align*}
Here, the second to last step uses that $\binom{a}{b}\leq (ae/b)^b$ and the final step exploits that $p\in o(n^{-1/(f+1)})$.

For the second claim, assume w.l.o.g.\ $p\leq 1/3$, as increasing $p$ further certainly increases the probability of the system to fail. For any $v\in V$, the probability that $|\{v_i\in F'\}|> f$ is independent of the same event for other nodes and larger than
\begin{equation*}
\binom{2f+1}{f+1}p^{f+1}(1-p)^f\geq \left(\frac{3}{2}\right)^f p^{f+1}(1-p)^f\geq p^{f+1},
\end{equation*}
since $\binom{a}{b}\geq (a/b)^b$ and $1-p\geq 2/3$. Hence, if $G$ contains $\Omega(n)$ nodes $v$ with non-zero outdegree, $p\in \omega(n^{-1/(f+1)})$ implies that the probability that there is such a node $v$ for which $|\{v_i\in F'\}|> f$ is at least
\begin{equation*}
1-\left(1-p^{f+1}\right)^{\Omega(n)}\subseteq 1-\left(1-\omega\left(\frac{1}{n}\right)\right)^{\Omega(n)}= 1-o(1).
\end{equation*}
If there is such a node $v$, there are algorithms $A$ and inputs so that $A$ sends a message across some edge $(v,w)$ in some round. If faulty nodes do not send messages in this round, the nodes $w_i\in V'\setminus F'$ do not receive the correct message from more than $f$ nodes $v_i$ and the simulation fails. Hence, the reinforcement cannot be valid.
\end{proof}
\begin{rem}
For constant $p$, one can determine suitable values of $f\in \Theta(\log n)$ using Chernoff's bound. However, as our focus is on small (constant) overhead factors, we refrain from presenting the calculation here.
\end{rem}

%###
\paragraph{Efficiency of the Reinforcement}
%###
For $f\in \NN$, we have that $\nu = \ell = 2f+1$ and $\eta = \ell^2 = 4f^2 + 4f + 1$, while we can sustain $p\in o(n^{-1/(f+1)})$.
In the special case of $f=1$, we improve from $p\in o(1/n)$ for the original network to $p\in o(1/\sqrt{n})$ by tripling the number of nodes.
However, $\eta = 9$, i.e., while the number of edges also increases only by a constant, it seems too large in systems where the limiting factor is the amount of links that can be afforded.

\section{Strong Reinforcement under \texorpdfstring{\om{p}}{Om(p)}}\label{sec:strong_om}
The strong reinforcement from the previous section is, trivially, also a strong reinforcement under \om{p}. However, we can reduce the number of copies per node for the weaker fault model. Given are the input network $G=(V,E)$ and scheduling algorithm $A$. Fix a parameter $f\in \NN$ and, this time, set $\ell = f+1$.

%###
\paragraph{Reinforced Network $G'$}
%###
We set $V'\triangleq V\times [\ell]$ and denote $v_i\triangleq (v,i)$. Accordingly, $P(v_i)\triangleq v$. We define $E'\triangleq \{(v',w')\in V'\times V'\,|\,(P(v'),P(w'))\in E\}$.

%###
\paragraph{Strong Simulation $A'$ of $A$}
%###
Each node\footnote{Nodes suffering omission failures still can simulate $A$ correctly.} $v'\in V'$
\begin{compactenum}[\bfseries (1)]
\item initializes local copies of all state variables of $v$ as in $A$,
\item \sloppy sends on each link $(v',w')\in E'$ in each round the message $v$ would send on $(P(v'),P(w'))$ when executing $A$, and
\item for each neighbor $w$ of $P(v')$ and each round $r$, updates the local copy of the state of $A$ as if $v$ received the (unique) message that has been sent to $v'$ by some of the nodes $w'$ with $P(w')=w$ (each one using edge $(w',v')$).
\end{compactenum}
Naturally, the last step assumes that some such neighbor sends a message and all $w'$ with $P(w')$ send the same such message; otherwise, the simulation fails. We show that $A'$ can be executed and simulates $A$ provided that for each $v\in V$, no more than $f$ of its copies are in $F'$.
\begin{lemma}\label{lemma:sim_om}
If for each $v\in V$, $|\{v_i\in F'\}|\leq f$, $A'$ strongly simulates $A$.
\end{lemma}
\begin{proof}
Analogous to the one of Lemma~\ref{lemma:sim_byz}, with the difference that faulty nodes may only omit sending messages and thus a single correct copy per node is sufficient.
\end{proof}

%###
\paragraph{Resilience of the Reinforcement}
%###
We now examine how large the probability $p$ can be for the precondition of Lemma~\ref{lemma:sim_byz} to be satisfied a.a.s.
\begin{theorem}
The above construction is a valid strong reinforcement for the fault model \om{p} if $p \in o(n^{-1/(f+1)})$. If $G$ contains $\Omega(n)$ nodes with non-zero outdegree, $p\in \omega(n^{-1/(f+1)})$ implies that the reinforcement is not valid.
\end{theorem}
\begin{proof}
By Lemma~\ref{lemma:sim_om}, $A'$ strongly simulates $A$ if for each $v\in V$, $|\{v_i\in F'\}|\leq f = \ell -1$. For $v\in V$,
\begin{equation*}
\Prr{\{v_i\,|\,i\in [\ell]\}\cap F'=\ell} = p^{f+1}.
\end{equation*}
By a union bound, $A'$ thus simulates $A$ with probability $1-o(1)$ if $p\in o(n^{-1/(f+1)})$.

Conversely, if there are $\Omega(n)$ nodes with non-zero outdegree and $p\in \omega(n^{-1/(f+1)})$, with probability $1-o(1)$ all copies of at least one such node $v$ are faulty. If $v$ sends a message under $A$, but all corresponding messages of copies of $v$ are not sent, the simulation fails. This shows that in this case the reinforcement is not valid.
\end{proof}
%###
\paragraph{Efficiency of the Reinforcement}
%###
For $f\in \NN$, we have that $\nu = \ell = f+1$ and $\eta = \ell^2 = f^2 + 2f + 1$, while we can sustain $p\in o(n^{-1/(f+1)})$.
In the special case of $f=1$, we improve from $p\in o(1/n)$ for the original network to $p\in o(1/\sqrt{n})$ by doubling the number of nodes and quadrupling the number of edges.

\section{More Efficient Reinforcement}\label{sec:eff}
In this section, we reduce the overhead in terms of edges at the expense of obtaining reinforcements that are not strong. We stress that the obtained trade-off between redundancy ($\nu$ and $\eta$) and the sustainable probability of faults $p$ is asymptotically optimal: as we require to preserve arbitrary routing schemes in a blackbox fashion, we need sufficient redundancy on the link level to directly simulate communication. From this observation, both for \om{p} and \byz{p} we can readily derive trivial lower bounds on redundancy that match the constructions below up to lower-order terms.

\subsection{A Toy Example}
Before we give the construction, we give some intuition on how we can reduce the number of required edges. Consider the following simple case. $G$ is a single path of $n$ vertices $(v_1,\ldots, v_n)$, and the schedule requires that in round $i$, a message is sent from $v_i$ to $v_{i+1}$. We would like to use a ``budget'' of only $n$ additional vertices and an additional $(1+\eps) m=(1+\eps) (n-1)$ links, assuming the fault model \om{p}. One approach is to duplicate the path and extend the routing scheme accordingly. We already used our entire budget apart from $\eps m$ links! This reinforcement is valid as long as one of the paths succeeds in delivering the message all the way.
The probability that one of the paths ``survives'' is
%\[
$1-(1-(1-p)^n)^2 \leq 1-(1-e^{-pn})^2 \leq e^{-2pn}$,
%\]
where we used that $1-x\leq e^{-x}$ for any $x\in \mathbb{R}$.
Hence, for any $p = \omega(1/n)$, the survival probability is $o(1)$. In contrast, the strong reinforcement with $\ell=2$ (i.e., $f=1$) given in \S~\ref{sec:strong_om} sustains any $p\in o(1/\sqrt{n})$ with probability $1-o(1)$; however, while it adds $n$ nodes only, it requires $3m$ additional edges.

We need to add some additional edges to avoid that the likelihood of the message reaching its destination drops too quickly. To this end, we use the remaining $\varepsilon m$ edges to ``cross'' between the two paths every $h\triangleq 2/\varepsilon$ hops (assume $h$ is an integer), cf.~Figure~\ref{fig:toy_and_grid}.
 This splits the path into segments of $h$ nodes each. As long as, for each such segment, in one of its copies all nodes survive, the message is delivered. For a given segment, this occurs with probability $1-(1-(1-p)^h)^2\geq 1-(ph)^2$. Overall, the message is thus delivered with probability at least $(1-(ph)^2)^{n/h}\geq 1-nhp^2$.
As for any constant $\varepsilon$, $h$ is a constant, this means that the message is delivered a.a.s.\ granted that $p\in o(1/\sqrt{n})$!

\begin{rem}
The reader is cautioned to not conclude from this example that random sampling of edges will be sufficient for our purposes in more involved graphs. Since we want to handle arbitrary routing schemes, we have no control over the number of utilized routing paths. As the latter is exponential in $n$, the probability that a fixed path is not ``broken'' by $F$ would have to be exponentially small in $n$. Moreover, trying to leverage Lov\'asz Local Lemma for a deterministic result runs into the problem that there is no (reasonable) bound on the number of routing paths that pass through a single node, i.e., the relevant random variables (i.e., whether a path ``survives'') exhibit lots of dependencies.
\end{rem}

\begin{figure}[H]
      \centering
        \includegraphics[width=\columnwidth,trim=0 00bp 0 000bp,clip]{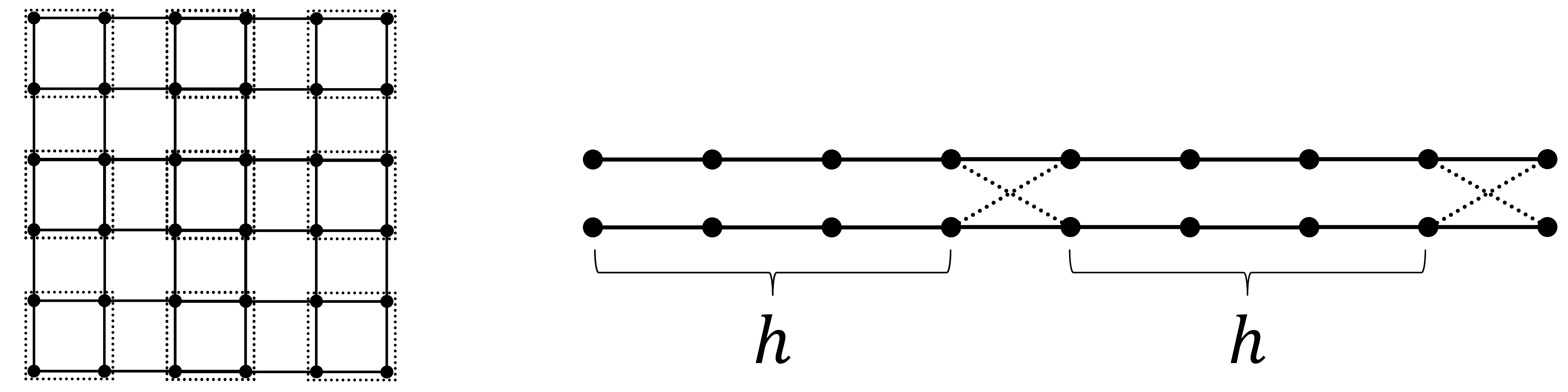}
      \caption{On the right: our toy example with $n=9$, $m=8$, $\varepsilon =1/2$, and $h=4$. The number of additional edges is $(1+\varepsilon)m$, instead of $3m$ as in the strong reinforcement construction. On the left: a $6$-ary $2$-dimensional hypercube. The subdivision of the node set into $2$-ary $2$-dimensional subcubes is illustrated by dotted lines.\label{fig:toy_and_grid}}
\end{figure}

\subsection{Partitioning the Graph}
To apply the above strategy to other graphs, we must take into account that there can be multiple intertwined routing paths. However, the key point in the above example was not that we had path segments, but rather that we partitioned the nodes into constant-size regions and added few edges inside these regions, while fully connecting the copies of nodes at the boundary of the regions.

In general, it is not possible to partition the nodes into constant-sized subsets such that only a very small fraction of the edges connects different subsets; any graph with good expansion is a counter-example. Fortunately, many network topologies used in practice are good candidates for our approach. In the following, we will discuss grid networks and minor free graphs, and show how to apply the above strategy in each of these families of graphs.

%###
\paragraph{Grid Networks}
%###
We can generalize the above strategy to hypercubes of dimension $d>1$.

\begin{definition}[Hypercube Networks]
A $q$-ary $d$-dimensional hypercube has node set $[q]^d$ and two nodes are adjacent if they agree on all but one index $i\in [d]$, for which $|v_i-w_i|=1$.
\end{definition}

\begin{lemma}\label{lemma:hypercube}
For any $h,d\in \NN$, assume that $h$ divides $q\in \NN$ and set $\varepsilon=1/h$. Then the $q$-ary $d$-dimensional hypercube can be partitioned into $(q/h)^d$ regions of $h^d$ nodes such that at most an $\varepsilon$-fraction of the edges connects nodes from different regions.
\end{lemma}
\begin{proof}
We subdivide the node set into $h$-ary $d$-dimensional subcubes; for an example of the subdivision of the node set of a $6$-ary $2$-dimensional hypercube into $2$-ary $2$-dimensional subcubes see Figure~\ref{fig:toy_and_grid}. There are $(q/h)^d$ such subcubes. The edges crossing the regions are those connecting the faces of adjacent subcubes. For each subcube, we attribute for each dimension one face to each subcube (the opposite face being accounted for by the adjacent subcube in that direction). Thus, we have at most $dh^{d-1}$ crossing edges per subcube. The total number of edges per subcube are these crossing edges plus the $d(h-1)h^{d-1}$ edges within the subcube. Overall, the fraction of crossedges is thus at most $1/(1+(h-1))=1/h$, as claimed.
\end{proof}

Note that the above result and proof extend to tori, which also include the ``wrap-around'' edges connecting the first and last nodes in any given dimension.

%###
\paragraph{Minor free Graphs}
%###
Another general class of graphs that can be partitioned in a similar fashion are minor free bounded-degree graphs.

\begin{definition}[$H$-Minor free Graphs]
For a fixed graph $H$, $H$ is a minor of $G$ if $H$ is isomorphic to a graph that can be obtained by zero or more
edge contractions on a subgraph of $G$. We say that a graph G is $H$-minor free if $H$ is not a minor of $G$.
\end{definition}
For any such graph, we can apply a corollary from \cite[Coro.~2]{LeviR15}, which is based on~\cite{alon1990separator}, to construct a suitable partition.

\begin{theorem}[\cite{LeviR15}]
Let $H$ be a fixed graph. There is a constant $c(H) > 1$ such that for every $\eps \in (0, 1]$ and
every $H$-minor free graph $G = (V, E)$ with degree bounded by $\Delta$ a partition $R_1,\ldots,R_k\subseteq V$ with the following properties can be found in time $O(|V|^{3/2})$:
\begin{enumerate}[(i)]
    \item $\forall i : |R_i|\leq \frac{c(H)\Delta^2}{\eps^2}$,
    \item $\forall i$ the subgraph induced by $R_i$ in $G$ is connected.
    \item $|\{(u,v) \mid u \in R_i, v \in R_j, i\neq j\}|\leq \eps \cdot |V|$.
\end{enumerate}
\end{theorem}
\begin{rem}
Grids and tori of dimension $d>2$ are not minor-free.
\end{rem}
We note that this construction is not satisfactory, as it involves large constants. It demonstrates that a large class of graphs is amenable to the suggested approach, but it is advisable to search for optimized constructions for more specialized graph families before applying the scheme.

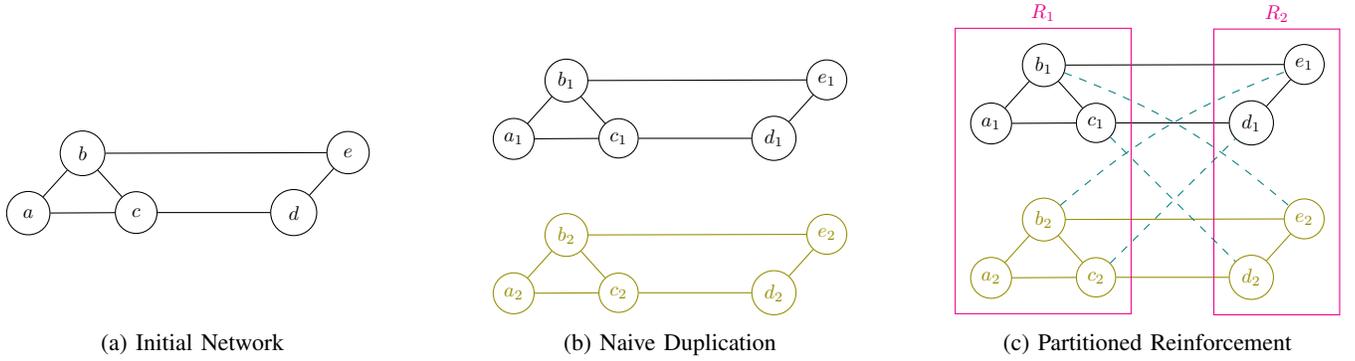
\begin{figure*}
\begin{minipage}[b]{0.3\textwidth}
%\vspace*{-30pt}
\begin{tikzpicture}[main/.style = {draw, circle}, scale=0.75, transform shape]
\node[main] (1) {$\ a\ $};
\node[main] (2) [above right=0.5cm and 0.4cm of 1] {$\ b\ $};
\node[main] (3) [below right=0.5cm and 0.4cm of 2] {$\ c\ $};
\node[main] (4) [right=2cm of 3] {$\ d\ $};
\node[main] (5) [above right=0.5cm and 0.4cm of 4] {$\ e\ $};
\draw (1) -- (2);
\draw (1) -- (3);
\draw (2) -- (3);
\draw (3) -- (4);
\draw (4) -- (5);
\draw (2) -- (5);
\end{tikzpicture}
\vspace*{30pt}
\centering\subcaption{Initial Network}
\end{minipage}\hfill
\begin{minipage}[b]{0.3\textwidth}
\begin{tikzpicture}[main/.style = {draw, circle}, scale=0.75, transform shape]
\node[main] (1) {$a_1$};
\node[main] (2) [above right=0.5cm and 0.4cm of 1] {$b_1$};
\node[main] (3) [below right=0.5cm and 0.4cm of 2] {$c_1$};
\node[main] (4) [right=2cm of 3] {$d_1$};
\node[main] (5) [above right=0.5cm and 0.4cm of 4] {$e_1$};
\node[main] (6) [olive, below=2cm of 1] {$a_2$};
\node[main] (7) [olive, above right=0.5cm and 0.4cm of 6] {$b_2$};
\node[main] (8) [olive, below right=0.5cm and 0.4cm of 7] {$c_2$};
\node[main] (9) [olive, right=2cm of 8] {$d_2$};
\node[main] (10) [olive, above right=0.5cm and 0.4cm of 9] {$e_2$};
\draw (1) -- (2);
\draw (1) -- (3);
\draw (2) -- (3);
\draw (3) -- (4);
\draw (4) -- (5);
\draw (2) -- (5);
\draw [olive] (6) -- (7);
\draw [olive] (6) -- (8);
\draw [olive] (7) -- (8);
\draw [olive] (8) -- (9);
\draw [olive] (9) -- (10);
\draw [olive] (7) -- (10);
\end{tikzpicture}
%\captionsetup{labelformat=empty}
\centering\subcaption{Naive Duplication}
\end{minipage}\hfill
\begin{minipage}[b]{0.3\textwidth}
\begin{tikzpicture}[main/.style = {draw, circle}, scale=0.75, transform shape]
\node[main] (1) {$a_1$};
\node[main] (2) [above right=0.5cm and 0.4cm of 1] {$b_1$};
\node[main] (3) [below right=0.5cm and 0.4cm of 2] {$c_1$};
\node[main] (4) [right=2cm of 3] {$d_1$};
\node[main] (5) [above right=0.5cm and 0.4cm of 4] {$e_1$};
\node[main] (6) [olive, below=2cm of 1] {$a_2$};
\node[main] (7) [olive, above right=0.5cm and 0.4cm of 6] {$b_2$};
\node[main] (8) [olive, below right=0.5cm and 0.4cm of 7] {$c_2$};
\node[main] (9) [olive, right=2cm of 8] {$d_2$};
\node[main] (10) [olive, above right=0.5cm and 0.4cm of 9] {$e_2$};
\draw (1) -- (2);
\draw (1) -- (3);
\draw (2) -- (3);
\draw (3) -- (4);
\draw (4) -- (5);
\draw (2) -- (5);
\draw [olive] (6) -- (7);
\draw [olive] (6) -- (8);
\draw [olive] (7) -- (8);
\draw [olive] (8) -- (9);
\draw [olive] (9) -- (10);
\draw [olive] (7) -- (10);
\draw [teal, dashed] (3) -- (9);
\draw [teal, dashed] (8) -- (4);
\draw [teal, dashed] (2) to [bend left=10] (10);
\draw [teal, dashed] (7) to [bend left=10] (5);
\node[draw, magenta, inner sep=2mm, label={[text=magenta]above:$R_1$}, fit=(2) (6) (3) (1)] {};
\node[draw, magenta, inner sep=2mm, label={[text=magenta]above:$R_2$}, fit=(5) (9) (5) (4)] {};
\end{tikzpicture}
%\captionsetup{labelformat=empty}
\centering\subcaption{Partitioned Reinforcement}
\end{minipage}
\caption{Comparison between a sample network, its naive duplication, and its reinforcement using two replications ($\ell=2$) and two partitions ($k=2$). The node overhead, edge overhead, and maximum node fault probability tolerance ($p$) under omission fault for 99\% network reliability for these three networks are (a) 1, 1, $\sim$0.002; (b) 2, 2, $\sim$0.02; and (c) 2, 2.67, $\sim$0.028, respectively. Note that both the naive duplication and the reinforced networks are guaranteed to be robust to one faulty node. However, the latter can handle some additional cases, like $c_1$ and $d_2$ nodes being faulty.}
\label{fig:net_compare}
\end{figure*}

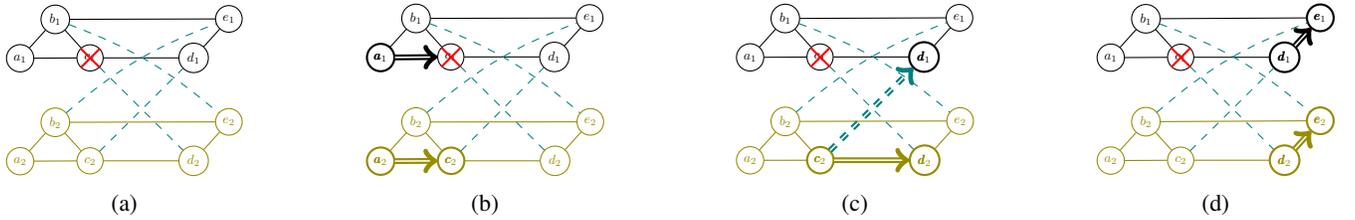
\begin{figure*}
\begin{minipage}[b]{0.2\textwidth}
\begin{tikzpicture}[main/.style = {draw, circle}, scale=0.5, transform shape]
\node[main] (1) {$a_1$};
\node[main] (2) [above right=0.5cm and 0.4cm of 1] {$b_1$};
\node[main] (3) [below right=0.5cm and 0.4cm of 2] {$c_1$};
\node[main, thick, cross out, outer sep=4pt] (_3) [red, below right=0.5cm and 0.4cm of 2] {};
\node[main] (4) [right=2cm of 3] {$d_1$};
\node[main] (5) [above right=0.5cm and 0.4cm of 4] {$e_1$};
\node[main] (6) [olive, below=2cm of 1] {$a_2$};
\node[main] (7) [olive, above right=0.5cm and 0.4cm of 6] {$b_2$};
\node[main] (8) [olive, below right=0.5cm and 0.4cm of 7] {$c_2$};
\node[main] (9) [olive, right=2cm of 8] {$d_2$};
\node[main] (10) [olive, above right=0.5cm and 0.4cm of 9] {$e_2$};
\draw (1) -- (2);
\draw (1) -- (3);
\draw (2) -- (3);
\draw (3) -- (4);
\draw (4) -- (5);
\draw (2) -- (5);
\draw [olive] (6) -- (7);
\draw [olive] (6) -- (8);
\draw [olive] (7) -- (8);
\draw [olive] (8) -- (9);
\draw [olive] (9) -- (10);
\draw [olive] (7) -- (10);
\draw [teal, dashed] (3) -- (9);
\draw [teal, dashed] (8) -- (4);
\draw [teal, dashed] (2) to [bend left=10] (10);
\draw [teal, dashed] (7) to [bend left=10] (5);
\end{tikzpicture}
%\captionsetup{labelformat=empty}
\centering\subcaption{}
\end{minipage}\hfill
\begin{minipage}[b]{0.2\textwidth}
\begin{tikzpicture}[main/.style = {draw, circle}, scale=0.5, transform shape]
\node[main, thick] (1) {$\pmb a_1$};
\node[main] (2) [above right=0.5cm and 0.4cm of 1] {$b_1$};
\node[main] (3) [below right=0.5cm and 0.4cm of 2] {$c_1$};
\node[main, thick, cross out, outer sep=4pt] (_3) [red, below right=0.5cm and 0.4cm of 2] {};
\node[main] (4) [right=2cm of 3] {$d_1$};
\node[main] (5) [above right=0.5cm and 0.4cm of 4] {$e_1$};
\node[main, thick] (6) [olive, below=2cm of 1] {$\pmb a_2$};
\node[main] (7) [olive, above right=0.5cm and 0.4cm of 6] {$b_2$};
\node[main, thick] (8) [olive, below right=0.5cm and 0.4cm of 7] {$\pmb c_2$};
\node[main] (9) [olive, right=2cm of 8] {$d_2$};
\node[main] (10) [olive, above right=0.5cm and 0.4cm of 9] {$e_2$};
\draw (1) -- (2);
\draw [thick, double, ->] (1) -- (3);
\draw (2) -- (3);
\draw (3) -- (4);
\draw (4) -- (5);
\draw (2) -- (5);
\draw [olive] (6) -- (7);
\draw [olive, thick, double, ->] (6) -- (8);
\draw [olive] (7) -- (8);
\draw [olive] (8) -- (9);
\draw [olive] (9) -- (10);
\draw [olive] (7) -- (10);
\draw [teal, dashed] (3) -- (9);
\draw [teal, dashed] (8) -- (4);
\draw [teal, dashed] (2) to [bend left=10] (10);
\draw [teal, dashed] (7) to [bend left=10] (5);
\end{tikzpicture}
%\captionsetup{labelformat=empty}
\centering\subcaption{}
\end{minipage} \hfill
\begin{minipage}[b]{0.2\textwidth}
\begin{tikzpicture}[main/.style = {draw, circle}, scale=0.5, transform shape]
\node[main] (1) {$a_1$};
\node[main] (2) [above right=0.5cm and 0.4cm of 1] {$b_1$};
\node[main] (3) [below right=0.5cm and 0.4cm of 2] {$c_1$};
\node[main, thick, cross out, outer sep=4pt] (_3) [red, below right=0.5cm and 0.4cm of 2] {};
\node[main, thick] (4) [right=2cm of 3] {$\pmb d_1$};
\node[main] (5) [above right=0.5cm and 0.4cm of 4] {$e_1$};
\node[main] (6) [olive, below=2cm of 1] {$a_2$};
\node[main] (7) [olive, above right=0.5cm and 0.4cm of 6] {$b_2$};
\node[main, thick] (8) [olive, below right=0.5cm and 0.4cm of 7] {$\pmb c_2$};
\node[main, thick] (9) [olive, right=2cm of 8] {$\pmb d_2$};
\node[main] (10) [olive, above right=0.5cm and 0.4cm of 9] {$e_2$};
\draw (1) -- (2);
\draw (1) -- (3);
\draw (2) -- (3);
\draw (3) -- (4);
\draw (4) -- (5);
\draw (2) -- (5);
\draw [olive] (6) -- (7);
\draw [olive] (6) -- (8);
\draw [olive] (7) -- (8);
\draw [olive, thick, double, ->] (8) -- (9);
\draw [olive] (9) -- (10);
\draw [olive] (7) -- (10);
\draw [teal, dashed] (3) -- (9);
\draw [teal, dashed, thick, double, ->] (8) -- (4);
\draw [teal, dashed] (2) to [bend left=10] (10);
\draw [teal, dashed] (7) to [bend left=10] (5);
\end{tikzpicture}
\centering\subcaption{}
\end{minipage}\hfill
\begin{minipage}[b]{0.2\textwidth}
\begin{tikzpicture}[main/.style = {draw, circle}, scale=0.5, transform shape]
\node[main] (1) {$a_1$};
\node[main] (2) [above right=0.5cm and 0.4cm of 1] {$b_1$};
\node[main] (3) [below right=0.5cm and 0.4cm of 2] {$c_1$};
\node[main, thick, cross out, outer sep=4pt] (_3) [red, below right=0.5cm and 0.4cm of 2] {};
\node[main, thick] (4) [right=2cm of 3] {$\pmb d_1$};
\node[main, thick] (5) [above right=0.5cm and 0.4cm of 4] {$\pmb e_1$};
\node[main] (6) [olive, below=2cm of 1] {$a_2$};
\node[main] (7) [olive, above right=0.5cm and 0.4cm of 6] {$b_2$};
\node[main] (8) [olive, below right=0.5cm and 0.4cm of 7] {$c_2$};
\node[main, thick] (9) [olive, right=2cm of 8] {$\pmb d_2$};
\node[main, thick] (10) [olive, above right=0.5cm and 0.4cm of 9] {$\pmb e_2$};
\draw (1) -- (2);
\draw (1) -- (3);
\draw (2) -- (3);
\draw (3) -- (4);
\draw [thick, double, ->] (4) -- (5);
\draw (2) -- (5);
\draw [olive] (6) -- (7);
\draw [olive] (6) -- (8);
\draw [olive] (7) -- (8);
\draw [olive] (8) -- (9);
\draw [olive, thick, double, ->] (9) -- (10);
\draw [olive] (7) -- (10);
\draw [teal, dashed] (3) -- (9);
\draw [teal, dashed] (8) -- (4);
\draw [teal, dashed] (2) to [bend left=10] (10);
\draw [teal, dashed] (7) to [bend left=10] (5);
\end{tikzpicture}
%\captionsetup{labelformat=empty}
\centering\subcaption{}
\end{minipage}
\caption{An illustration of a sample routing in the partitioned reinforced network from Figure \ref{fig:net_compare}. (a) Consider the case when node $c_1$ is faulty, under the omission fault assumption, and a message was to be sent along the $a-c-d-e$ route in the initial network. The goal is to track this message in the reinforced network. (b) The message is sent from $a_1$ to $c_1$, and from $a_2$ to $c_2$. (c) Node $c_1$ is not able to send the message to the next node. On the other hand, node $c_2$ sends the message to both $d_1$ and $d_2$. (d) $d_1$ and $d_2$ both receive a message, and therefore, send it to $e_1$ and $e_2$, respectively. At this point, the routing is successful, as only one of $e_1$ or $e_2$ receiving the message is sufficient. While this is an example of routing a single message, the reinforced network is able to operate for any algorithm that is runnable on the initial network. More detail can be found in Section \ref{subsec:part_reinforce}.}
\label{fig:sample_routing}
\end{figure*}

%###
\subsection{Reinforcement}
\label{subsec:part_reinforce}
%###
Equipped with a suitable partition of the original graph $G=(V,E)$ into disjoint regions $R_1,\ldots,R_k\subseteq V$, we reinforce as follows.
As before, we set $V'\triangleq V\times [\ell]$, denote $v_i\triangleq (v,i)$, define $P(v_i)\triangleq v$, and set $\ell\triangleq f+1$. However, the edge set of $G'$ differs. For $e=(v,w)\in E$,
\begin{align*}
E_e'\triangleq
 \begin{cases}
\{(v_i,w_i)\,|\,i\in [\ell]\} & \mbox{if $\exists k'\in [k]: v,w\in R_{k'}$}\\
\{(v_i,w_j)\,|\,i,j\in [\ell]\} & \mbox{else,}
\end{cases}
\end{align*}
and we set $E'\triangleq \bigcup_{e\in E} E_e'$.

%###
\paragraph{Simulation under \texorpdfstring{\om{p}}{Om(p)}}
%###
Consider $v\in V$. We want to maintain the invariant that in each round, \emph{some} $v_i$ has a copy of the state of $v$ in $A$. To this end, $v'\in V'$
\begin{compactenum}[\bfseries (1)]
\item initializes local copies of all state variables of $v$ as in $A$ and sets $\know_{v'}=\true$;
\item sends on each link $(v',w')\in E'$ in each round
\begin{compactitem}
\item message $M$, if $P(v')$ would send $M$ via $(P(v'),P(w'))$ when executing $A$ and $\know_{v'}=\true$,
\item a special symbol $\bot$ if $\know_{v'}=\true$, but $v$ would not send a message via $(P(v'),P(w'))$ according to $A$, or
\item no message if $\know_{v'}=\false$;
\end{compactitem}
\item if, in a given round, $\know_{v'}=\true$ and $v'$ receives for each neighbor $w$ of $P(v')$ a message from some $w_j\in V'$, it updates the local copy of the state of $v$ in $A$ as if $P(v')$ received this message (interpreting $\bot$ as no message); and
\item if this is not the case, $v'$ sets $\know_{v'}=\false$.
\end{compactenum}
We claim that as long as $\know_{v'}=\true$ at $v'$, $v'$ has indeed a copy of the state of $P(v')$ in the corresponding execution of $A$; therefore, it can send the right messages and update its state variables correctly.

\begin{lemma}\label{applemma:weak_om}
Suppose that for each $k'\in [k]$, there is some $i\in [\ell]$ so that $\{v_i\,|\,v\in R_{k'}\}\cap F'=\emptyset$. Then $A'$ simulates $A$.
\end{lemma}
\begin{proof}
Select for each $R_{k'}$, $k'\in [k]$, some $i$ such that $\{v_i\,|\,v\in R_{k'}\}\cap F'=\emptyset$ and denote by $C$ the union of all these nodes. As $P(C)=V$, it suffices to show that each $v'\in C$ successfully maintains a copy of the state of $P(v')$ under $A$. However, we also need to make sure that \emph{all} messages, not only the ones sent by nodes in $c$, are ``correct,'' in the sense that a message sent over edge $(v',w')\in E'$ in round $r$ would be sent by $A$ over $(P(v'),P(w'))$ (where $\bot$ means no message is sent). Therefore, we will argue that the set of nodes $T_r\triangleq \{v'\in V'\,|\,\know_{v'}=\true \text{ in round }r\}$ knows the state of their counterpart $P(v')$ under $A$ up to and including round $r\in \NN$. As nodes $v'$ with $\know_{v'}=\false$ do not send any messages, this invariant guarantees that all sent messages are correct in the above sense.

We now show by induction on the round number $r\in \NN$ that (i) each $v'\in T_r$ knows the state of $P(v')$ under $A$ and (ii) $C\subseteq T_r$. Due to initialization, this is correct initially, i.e., in ``round $0$;'' we use this to anchor the induction at $r=0$, setting $T_0\triangleq V'$.

For the step from $r$ to $r+1$, note that because all $v'\in T_r$ have a copy of the state of $P(v')$ at the end of round $r$ by the induction hypothesis, each of them can correctly determine the message $P(v')$ would send over link $(v,w)\in E$ in round $r+1$ and send it over each $(v',w')\in E'$ with $P(w')=w$. Recall that $v'\in T_{r+1}$ if and only if $v'\in T_r$ and for each $(w,P(v'))\in E$ there is at least one $w'\in V'$ with $P(w')=w$ from which $v'$ receives a message. Since under \om{p} nodes in $F'$ may only omit sending messages, it follows that $v'\in T_{r+1}$ correctly updates the state variables of $P(v')$, just as $P(v')$ would in round $r+1$ of $A$.

It remains to show that $C\subseteq T_{r+1}$. Consider $v_i\in C$ and $(w,v)\in E$. If $v,w\in R_{k'}$ for some $k'\in [k]$, then $w_i\in C$ by definition of $C$. Hence, by the induction hypothesis, $w_i\in T_r$, and $w_i$ will send the message $w$ would send in round $r+1$ of $A$ over $(w,v)\in E$ to $v_i$, using the edge $(w_i,v_i)\in E'$. If this is not the case, then there is some $j\in [\ell]$ such that $w_j\in C$ and we have that $(w_j,v_i)\in E'$. Again, $v_i$ will receive the message $w$ would send in round $r+1$ of $A$ from $w_j$. We conclude that $v_i$ receives at least one copy of the message from $w$ for each $(w,v)\in E$, implying that $v\in T_{r+1}$ as claimed. Thus, the induction step succeeds and the proof is complete.
\end{proof}

Figure \ref{fig:net_compare} provides an example of a comparison between a network, a naive duplication of that network, and its reinforcement. The simulation process of sending a message in the same sample network is shown in Figure \ref{fig:sample_routing}.

%###
\paragraph{Resilience of the Reinforcement}
%###
We denote $R\triangleq \max_{k'\in [k]}\{|R_{k'}|\}$ and $r\triangleq \min_{k'\in [k]}\{|R_{k'}|\}$.
\begin{theorem}\label{thm:weak_om}
The above construction is a valid reinforcement for \om{p} if $p \in o((n/r)^{-1/(f+1)}/R)$. Moreover, if $G$ contains $\Omega(n)$ nodes with non-zero outdegree and $R\in O(1)$, $p\in \omega(n^{-1/(f+1)})$ implies that the reinforcement is not valid.
\end{theorem}
\begin{proof}
By Lemma~\ref{applemma:weak_om}, $A'$ simulates $A$ if for each $k'\in [k]$, there is some $i\in [\ell]$ so that $\{v_i\,|\,v\in R_{k'}\}\cap F'=\emptyset$. For fixed $k'$ and $i\in [\ell]$,
\begin{equation*}
\Prr{\{v_i\,|\,v\in R_{k'}\}\cap F'=\emptyset}=(1-p)^{|R_{k'}|}\geq 1-Rp.
\end{equation*}
Accordingly, the probability that for a given $k'$ the precondition of the lemma is violated is at most $(Rp)^{f+1}$. As $k\leq n/r$, taking a union bound over all $k'$ yields that with probability at least $1-n/r\cdot (Rp)^{f+1}$, $A'$ simulates $A$. Therefore, the reinforcement is valid if $p \in o((n/r)^{-1/(f+1)}/R)$.

Now assume that $r\leq R\in O(1)$ and also that $p\in \omega(n^{-1/(f+1)})\subseteq \omega((n/r)^{-1/(f+1)}/R)$. Thus, for each $v\in V$, all $v'\in V'$ with $P(v')=v$ simultaneously end up in $F'$ with probability $\omega(1/n)$. Therefore, if $\Omega(n)$ nodes have non-zero outdegree, with a probability in $1-(1-\omega(1/n))^{\Omega(n)}=1-o(1)$ for at least one such node $v$ all its copies end up in $F'$. In this case, the simulation fails if $v$ sends a message under $A$, but all copies of $v'$ suffer omission failures in the respective round.
\end{proof}
%###
\paragraph{Efficiency of the Reinforcement}
%###
For $f\in \NN$, we have that $\nu = \ell = f+1$ and $\eta = (1-\varepsilon)\ell + \varepsilon \ell^2 = 1+(1+\varepsilon)f+\varepsilon f^2$, while we can sustain $p\in o(n^{-1/(f+1)})$.
In the special case of $f=1$ and $\varepsilon=1/5$, we improve from $p\in o(1/n)$ for the original network to $p\in o(1/\sqrt{n})$ by doubling the number of nodes and multiplying the number of edges by~$2.4$.

\begin{rem}
For hypercubes and tori, the asymptotic notation for $p$ does not hide huge constants.
Lemma~\ref{lemma:hypercube} shows that $h$ enters the threshold in Theorem~\ref{thm:weak_om} as $h^{-d+1/2}$.
For the cases of $d=2$ and $d=3$, which are the most typical (for $d>3$ grids and tori suffer from large distortion when embedding them into $3$-dimensional space), the threshold on $p$ degrades by factors of $11.2$ and $55.9$, respectively.
\end{rem}

\subsection{Simulation under \texorpdfstring{\byz{p}}{Byz(p)}}

The same strategy can be applied for the stronger fault model \byz{p}, if we switch back to having $\ell=2f+1$ copies and nodes accepting the majority message among all messages from copies of a neighbor in the original graph.

Consider node $v\in V$. We want to maintain the invariant that in each round, a majority among the nodes $v_i$, $i\in [\ell]$, has a copy of the state of $v$ in $A$. For $v'\in V'$ and $(w,P(v'))\in E$, set $N_{v'}(w)\triangleq \{w'\in V'\,|\,(w',v')\in E'\}$. With this notation, $v'$ behaves as follows.
\begin{compactenum}[\bfseries (1)]
\item It initializes local copies of all state variables of $v$ as in $A$.
\item It sends in each round on each link $(v',w')\in E'$ the message $v$ would send on $(P(v'),P(w'))$ when executing $A$ (if $v'$ cannot compute this correctly, it may send an arbitrary message).
\item It updates its state in round $r$ as if it received, for each $(w,P(v'))\in E$, the message the majority of nodes in $N_{v'}(w)$ sent.
\end{compactenum}

\begin{lemma}\label{lemma:weak_byz}
Suppose for each $k'\in [k]$, there are at least $f+1$ indices $i\in [\ell]$ so that $\{v_i\,|\,v\in R_{k'}\}\cap F'=\emptyset$. Then $A'$ simulates $A$.
\end{lemma}
\begin{proof}
Select for each $R_{k'}$, $k'\in [k]$, $f+1$ indices $i$ such that $\{v_i\,|\,v\in R_{k'}\}\cap F'=\emptyset$ and denote by $C$ the union of all these nodes. We claim that each $v'\in C$ successfully maintains a copy of the state of $P(v')$ under $A$. We show this by induction on the round number $r\in \NN$, anchored at $r=0$ due to initialization.

For the step from $r$ to $r+1$, observe that because all $v'\in C$ have a copy of the state of $P(v')$ at the end of round $r$ by the induction hypothesis, each of them can correctly determine the message $P(v')$ would send over link $(v,w)\in E$ in round $r+1$ and send it over each $(v',w')\in E$ with $P(w')=w$. For each $v'\in C$ and each $(w,P(v'))$, we distinguish two cases. If $P(v')$ and $w$ are in the same region, let $i$ be such that $v'=v_i$. In this case, $N_{v'}(w)=\{w_i\}$ and, by definition of $C$, $w_i\in C$. Thus, by the induction hypothesis, $w_i$ sends the correct message in round $r+1$ over the link $(w',v')$. On the other hand, if $P(v')$ and $w$ are in different regions, $N_{v'}(w)=\{w_i\,|\,i\in [\ell]\}$. By the definition of $C$ and the induction hypothesis, the majority of these nodes (i.e., at least $f+1$ of them) sends the correct message $w$ would send over $(w,P(v'))$ in round $r+1$ when executing $A$. We conclude that $v'$ correctly updates its state, completing the proof.
\end{proof}

%###
\paragraph{Resilience of the Reinforcement}
%###
As before, denote $R\triangleq \max_{k'\in [k]}\{|R_{k'}|\}$ and $r\triangleq \min_{k'\in [k]}\{|R_{k'}|\}$.
\begin{theorem}\label{thm:byz}
The above construction is a valid reinforcement for the fault model \byz{p} if $p \in o((n/r)^{-1/(f+1)}/R)$. Moreover, if $G$ contains $\Omega(n)$ nodes with non-zero outdegree, $p\in \omega(n^{-1/(f+1)})$ implies that the reinforcement is not valid.
\end{theorem}
\begin{proof}
By Lemma~\ref{lemma:weak_byz}, $A'$ simulates $A$ if for each $k'\in [k]$, there are at least $f+1$ indices $i\in [\ell]$ so that $\{v_i\,|\,v\in R_{k'}\}\cap F'=\emptyset$. For fixed $k'$ and $i\in [\ell]$,
\begin{equation*}
\Prr{\{v_i\,|\,v\in R_{k'}\}\cap F'=\emptyset}=(1-p)^{|R_{k'}|}\geq 1-Rp.
\end{equation*}
Thus, analogous to the proof of Theorem~\ref{thm:strong_byz}, the probability that for a given $k'$ the condition is violated is at most
\begin{align*}
&\sum_{j=f+1}^{2f+1}\binom{2f+1}{j}(Rp)^j(1-Rp)^{2f+1-j}\\
=\,& (2e)^f(Rp)^{f+1}(1+o(1)).
\end{align*}
By a union bound over the at most $n/r$ regions, we conclude that the precondition $p \in o((n/r)^{-1/(f+1)}/R)$ guarantees that the simulation succeeds a.a.s.

For the second statement, observe that for each node $v\in V$ of non-zero outdegree,
\begin{equation*}
\Prr{|\{v_i\}\cap F'|\geq f+1}\geq p^{f+1}= \omega\left(\frac{1}{n}\right).
\end{equation*}
Thus, a.a.s.\ there is such a node $v$. Let $(v,w)\in E$ and assume that $A$ sends a message over $(v,w)$ in some round. If $v$ and $w$ are in the same region, the faulty nodes sending an incorrect message will result in a majority of the $2f+1=|\{w'\in V'\,|\,P(w')=w\}|$ copies of $w$ attaining an incorrect state (of the simulation), i.e., the simulation fails. Similarly, if $w$ is in a different region than $v$, for each copy of $w$ the majority message received from $N_{w'}(v)$ will be incorrect, resulting in an incorrect state.
\end{proof}
\begin{rem}
Note that the probability bounds in Theorem~\ref{thm:byz} are essentially tight in case $R\in O(1)$. A more careful analysis establishes similar results for $r\in \Theta(R)\cap \omega(1)$, by considering w.l.o.g.\ the case that all regions are connected and analyzing the probability that within a region, there is some path so that for at least $f+1$ copies of the path in $G'$, some node on the path is faulty. However, as again we consider the case $R\in O(1)$ to be the most interesting one, we refrain from generalizing the analysis.
\end{rem}
%###
\paragraph{Efficiency of the Reinforcement}
%###
For $f\in \NN$, we have that $\nu = \ell = 2f+1$ and $\eta = (1-\varepsilon)\ell + \varepsilon \ell^2 = 1+(2+2\varepsilon)f+4\varepsilon f^2$, while we can sustain $p\in o(n^{-1/(f+1)})$.
In the special case of $f=1$ and $\varepsilon=1/5$, we improve from $p\in o(1/n)$ for the original network to $p\in o(1/\sqrt{n})$ by tripling the number of nodes and multiplying the number of edges by~$4.2$.

\section{Empirical Evaluation}\label{sec:eval}

We have shown that our approach from \S~\ref{sec:eff} works particularly well
for graphs that admit a certain partitioning, such as
sparse graphs (e.g., minor-free graphs) or low-dimensional
hypercubes. To provide some empirical motivation for the relevance
of these examples, we note that the topologies collected
in the Rocketfuel \cite{spring2002measuring} and Internet Topology Zoo \cite{knight2011internet} projects
are all sparse: almost a third (namely 32\%) of the topologies even belong to the family of
cactus graphs, and roughly half of the graphs (49\%) are outerplanar \cite{sigmetrics18tomography}.

To complement our analytical results and study the reinforcement cost
of our approach in realistic networks, we conducted simulations on
the around 250 networks from the Internet Topology Zoo.
While we have a fairly good understanding of the different network topologies
deployed in practice, unfortunately, little is known about the state-of-the-art protection mechanisms used by network operators today. Network operators are typically reluctant to share details about their  infrastructure for security reasons, rendering a comparative evaluation difficult. That said, it seems relatively safe to assume that the most robust solutions rely on an one-by-one (``A/B'') replication strategy which allows to completely reroute traffic to a backup network; this baseline requires doubling resources and can hence be fairly costly.

In the following, we will report on our main insights.
Due to space constraints, we focus on the case of omission faults;
the results for Byzantine faults follow the same general trends.

Recall that we replace each node by $f+1$ of its copies, and each edge with endpoints in
different regions of the partition with $(f+1)^{2}$ copies; every other edge is replaced by $f+1$ copies.
Our goal is to do this partitioning such that it minimizes the edge overhead of the new network and
maximizes the probability of the network's resilience.
The fault probability of the network for given $p$, $f$ and partitions with $l_{1}, l_{2}, ..., l_{k}$ nodes is calculated as
$1 - \prod_{i=1}^{k} [1-(1-(1-p)^{l_{i}})^{f+1}]$.

In the following, as a case study, we fix a target network failure probability of at most $0.01$.
That is, the reinforced network is guaranteed to operate correctly with a probability of $99\%$, and we aim to maximize the probability $p$ with which nodes independently fail subject to this constraint.
For this fixed target resilience of the network, we determine the value of $p$ matching it using the above formula.
We remark that the qualitative behavior for smaller probabilities of network failure is the same, where the more stringent requirement means that our scheme outperforms naive approaches for even smaller network sizes.

For the examined topologies, it turned out that no specialized tools were needed to find good partitionings.
We considered a \emph{Spectral Graph Partitioning} tool~\cite{Hagen1992NewSM} and Metis \cite{Karypis1998fast},
a partitioning algorithm from a python library.
For small networks (less than 14 nodes), we further implemented a brute-force algorithm,
which provides an optimal baseline.

\begin{figure}[t]
	\centering
	\makebox[0pt]{\includegraphics[width=0.5\textwidth,trim=0 00bp 0 000bp,clip]{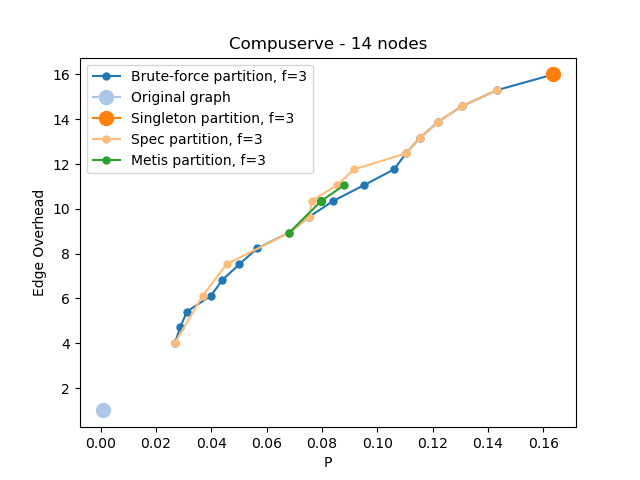}}
	\caption{Edge overhead for $f=3$ and different partitioning algorithms as a function of $p$.}
	\label{fig:compuserve}
\end{figure}

Figure \ref{fig:compuserve} shows the resulting edge overheads for the different partitioning algorithms
as a function of $p$ and for $f=3$, at hand of a specific example.
For reference, we added the value of $p$ for the original graph ($f=0$) to the plot, which has an overhead factor of $1$ (no redundancy).

As to be expected, for each algorithm and the fixed value of $f=3$, as the number of components in partitionings increases, the edge overhead and $p$
increase as well.
The ``Singleton partition'' point for $f=3$ indicates the extreme case where the size of the components is equal to 1 and the approach becomes identical to strong reinforcement (see \S~\ref{sec:strong_om});
hence, it has an edge overhead of $(f+1)^{2}=16$.
The leftmost points of the $f=3$ curves correspond to the other extreme of ``partitioning'' the nodes into a single set, resulting in naive replication of the original graph, at an edge overhead of $f+1=4$.

We observed this general behavior for networks of all sizes under varying $f$, where the spectral partitioning consistently outperformed Metis, and both performed very close to the brute force algorithm on networks to which it was applicable.
We concluded that the spectral partitioning algorithm is sufficient to obtain results that are close to optimal for the considered graphs, most of which have fewer than 100 nodes, with only a handful of examples with size between 100 and 200.
Accordingly, in the following we confine the presentation to the results obtained using the spectral partitioning algorithm.

\begin{figure}[t]
	\centering
	\makebox[0pt]{\includegraphics[width=0.5\textwidth,trim=0 00bp 0 000bp,clip]{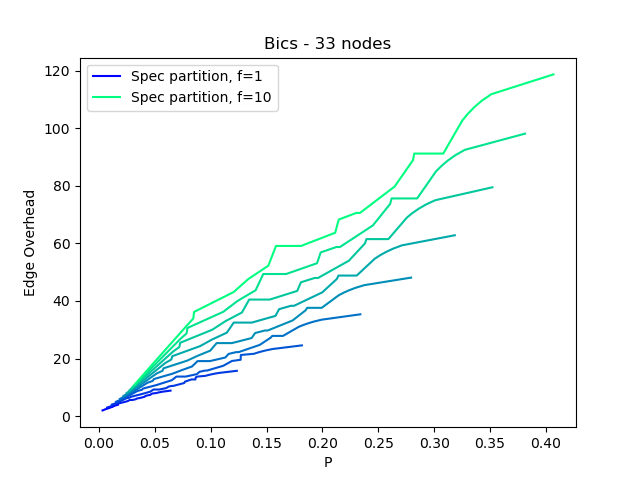}}
	\caption{Edge overhead for $f\in [1,10]$ as a function of $p$.}
	\label{fig:bisc10}
\end{figure}

In Figure \ref{fig:bisc10}, we take a closer look on how the edge overhead
depends on $f$, at hand of a network of 33 nodes. Note that the partitionings do not depend on $f$, causing the 10 curves to have similar shape.
As $f$ increases, the node overhead, edge overhead, and $p$ for the reinforced networks increase.
We can see that it is advisable to use larger values of $f$ only if the strong reinforcement approach for smaller $f$ cannot push $p$ to the desired value.
We also see that $f=1$ is sufficient to drive $p$ up to more than $6\%$, improving by almost two orders of magnitude over the roughly $0.01/33\approx 0.03\%$ the unmodified network can tolerate with probability $99\%$.
While increasing $f$ further does increase resilience, the relative gains are much smaller, suggesting that $f=1$ is the most interesting case.

\begin{figure}[t]
	\centering
	\makebox[0pt]{\includegraphics[width=0.5\textwidth,trim=0 00bp 0 000bp,clip]{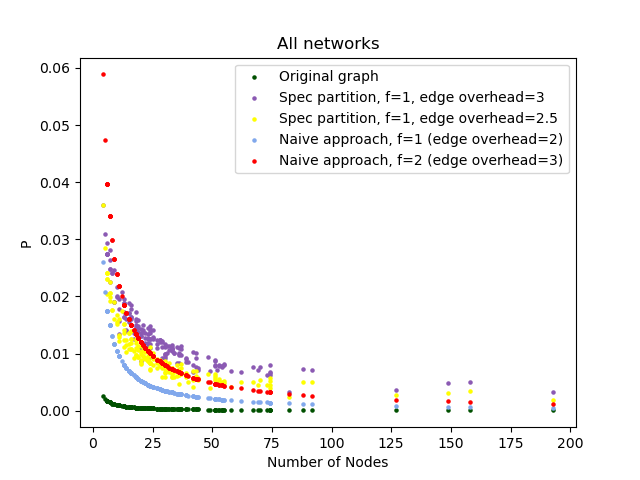}}
	\caption{Study of $p$ for all Topology Zoo networks and $f=1$, sorted by size.\label{fig:allnetworks}}
\end{figure}

Following up on this, in Figure \ref{fig:allnetworks} we plot $p$ for all existing networks in the Topology Zoo using the spectral graph partitioning algorithm and $f=1$.
Specifically, for each network, we calculated the value of $p$ on a set of reinforced networks with different node and edge overheads. Naturally, with increasing network size, the value of $p$ that can be sustained at a given overhead becomes smaller. Note, however, that naive replication quickly loses ground as $n$ becomes larger. In particular, already for about 20 nodes, an edge overhead of 3 with our approach is better than adding \emph{two} redundant copies of the original network, resulting in more nodes, but the same number of edges. Beyond roughly 50 nodes, our approach outperforms two independent copies of the network using fewer edges, i.e., an edge overhead of 2.5.

\begin{figure}[t]
	\centering
	\makebox[0pt]{\includegraphics[width=0.5\textwidth,trim=0 00bp 0 000bp,clip]{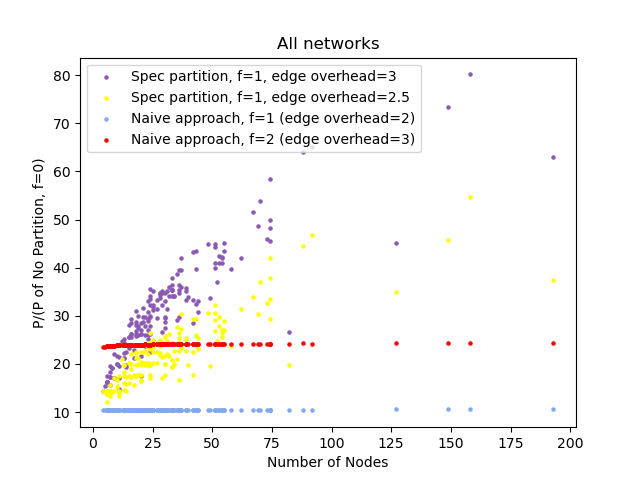}}
	\caption{Relative improvement over baseline for all Topology Zoo networks.\label{allnetworksratio}}
\end{figure}

To show more clearly when our approach outperforms naive network replication, Figure \ref{allnetworksratio} plots the \emph{relative gain} in the probability $p$ of node failure that can be sustained compared to the original network.

This plot is similar to the previous one. The y-axis now represents $p$ divided by the value of $p$ for the original graph. We now see that naive replication provides an almost constant improvement across the board. This is due to the fact that under this simple scheme, the reinforcement fails as soon as in each copy of the graph at least one node fails, as it is possible that a routing path in the original graph involves all nodes corresponding to failed copies.

Denote by $p_k$ the probability of node failure that can be sustained with $99\%$ reliability when simply using $k$ copies of the original graph (in particular $p_1\approx 0.01/n$). For small $k$, the probability $(1-p_k)^n$ that a single copy of the original graph is fault-free needs to be close to $1$. Hence, we can approximate $(1-p_k)^n\approx 1-p_k n$. The probability that all copies contain a failing node is hence approximately $(p_kn)^k$. Thus, $p_1 n \approx 0.01\approx (p_k n)^k$, yielding that
\begin{equation*}
\frac{p_k}{p_1}=\frac{p_k n}{p_1 n}\approx \frac{0.01^{1/k}}{0.01}=100^{1-1/k}.
\end{equation*}
In particular, we can expect ratios of roughly $10$ for $k=2$ and $21.5$ for $k=3$, respectively. The small discrepancy to the actual numbers is due to the approximation error, which would be smaller for higher target resilience.

As the plot clearly shows, our method achieves a relative improvement that increases with $n$, as predicted by Theorem~\ref{thm:weak_om}.
In conclusion, we see that our approach promises substantial improvements over the naive replication strategy,
which is commonly employed in mission-critical networks
(e.g., using dual planes as in RFC 7855~\cite{springpsr}).

\section{Discussion}\label{sec:disc}

In the previous sections, we have established that constant-factor redundancy can significantly increase reliability of the communication network in a blackbox fashion. Our constructions in \S~\ref{sec:eff} are close to optimal. Naturally, one might argue that the costs are still too high. However, apart from pointing out that the costs of using sufficiently reliable components may be even higher, we would like to raise a number of additional points in favor of the approach.

%###
\paragraph{Node Redundancy}
%###
When building reliable large-scale systems, fault-tolerance needs to be considered on all system levels. Unless nodes are sufficiently reliable, node replication is mandatory, regardless of the communication network. In other words, the node redundancy required by our construction may not be an actual overhead to begin with. When taking this point of view, the salient question becomes whether the increase in links is acceptable. Here, the first observation is that any system employing node redundancy will need to handle the arising additional communication, incurring the respective burden on the communication network. Apart from still having to handle the additional traffic, however, the system designer now needs to make sure that the network is sufficiently reliable for the node redundancy to matter. Our simple schemes then provide a means to provide the necessary communication infrastructure without risking to introduce, e.g., a single point of failure during the design of the communication network; at the same time, the design process is simplified and modularized.

%###
\paragraph{Dynamic Faults}
%###
Because of the introduced fault-tolerance, faulty components do not impede the system as a whole, so long as the simulation of the routing scheme can still be carried out. Hence, one may repair faulty nodes at runtime. If $T$ is the time for detecting and fixing a fault, we can discretize time in units of $T$ and denote by $p_T$ the (assumed to be independent) probability that a node is faulty in a given time slot, which can be bounded by twice the probability to fail within $T$ time. Then the failure probabilities we computed in our analysis directly translate to an upper bound on the expected fraction of time during which the system is not (fully) operational.

%###
\paragraph{Adaptivity}
%###
The employed node- and link-level redundancy may be required for mission-critical applications only, or the system may run into capacity issues. In this case, we can exploit that the reinforced network has a very simple structure, making various adaptive strategies straightforward to implement.
\begin{enumerate}[(i)]
  \item One might use a subnetwork only, deactivating the remaining nodes and links, such that a reinforced network for smaller $f$ (or a copy of the original network, if $f=0$) remains. This saves energy.
  \item One might subdivide the network into several smaller reinforced networks, each of which can perform different tasks.
  \item One might leverage the redundant links to increase the overall bandwidth between (copies of) nodes, at the expense of reliability.
  \item The above operations can be applied locally; e.g., in a congested region of the network, the link redundancy could be used for additional bandwidth. Note that if only a small part of the network is congested, the overall system reliability will not deteriorate significantly.
\end{enumerate}
Note that the above strategies can be refined and combined according to the profile of requirements of the system.

\section{Related Work}
\label{sec:relwork}

Robust routing is an essential feature of dependable
communication networks, and has been explored
intensively in the literature already.

%###
\paragraph*{Resilient Routing on the Network Layer}
%###
In contrast to our approach,
existing resilient routing mechanisms on the network layer
are typically \emph{reactive}.
They
can be categorized
according to whether they are supported in the
control plane, e.g.,
\cite{DBLP:conf/spaa/BuschST03,corson1995distributed,francois2005achieving,gafni-lr,greenberg2005clean,oran1990rfc1142},
or in the data plane, e.g.,~\cite{purr,srds19failover,ref4,ddc,plinko-full,keep-fwd},
see also the recent survey  \cite{frr-survey}.
These mechanisms are usually designed to cope with link failures.
Resilient routing algorithms in the control plane
typically rely on a global recomputation of paths
(either
centralized \cite{vahdat2015purpose},
distributed \cite{francois2005achieving}
or both \cite{icnp15shear}),
or on techniques based on link reversal \cite{gafni-lr}, and can
hence re-establish policies relatively easily;
however, they come at the price of a relatively high restoration time
\cite{francois2005achieving}.
Resilient routing algorithms in the dataplane can react to failures
significantly faster~\cite{feigenbaum2012brief}; however,
due to the local nature of the failover, it is challenging to
maintain network policies or even a high degree of resilience~\cite{chiesa2016resiliency}.
In this line of literature,
the network is usually given and the goal is to re-establish
routing paths quickly, ideally as long as the underlying physical
network is connected (known as perfect resilience~\cite{feigenbaum2012brief,disc20}).

In contrast, in this paper we ask the question of how to proactively enhance the
network in order to tolerate failures, rather than reacting to them. In particular, we consider more general failures,
beyond link failures and benign faults.
We argue that such a re-enforced
network simplifies routing as it is not necessary to compute new paths.
The resulting problems are very different in nature, also in terms
of the required algorithmic techniques.

%###
\paragraph*{Local Faults}
%###
In this paper, we consider more general failure models
than typically studied in the resilient routing literature above,
as our model is essentially a local fault model.
Byzantine faults were studied in~\cite{dolev2008constant,pelc2005broadcasting} in the context of broadcast and consensus problems. Unlike its global classical counterpart, the $f$-local Byzantine adversary can control at most $f$ neighbors of each vertex. This more restricted adversary gives rise to more scalable solutions, as the problems can be solved in networks of degree $O(f)$; without this restriction, degrees need to be proportional to the \emph{total} number of faults in the network.

We also limit our adversary in its selection of Byzantine nodes, by requiring that the faulty nodes are chosen independently at random. As illustrated, e.g., by Lemma~\ref{lemma:sim_byz} and Theorem~\ref{thm:strong_byz}, there is a close connection between the two settings. Informally, we show that certain values of $p$ correspond, asymptotically almost surely (a.a.s), to an $f$-local Byzantine adversary. However, we diverge from the approach in~\cite{dolev2008constant,pelc2005broadcasting} in that we require a fully time-preserving simulation of a fault-free routing schedule, as opposed to solving the routing task in the reinforced network from scratch.

%###
\paragraph*{Fault-Tolerant Logical Network Structures}
%###
Our work is reminiscent of literature on
the design fault-tolerant network structures.
In this area (see~\cite{Parter16} for a survey), the goal is to compute a sub-network that has a predefined property, e.g., containing minimum spanning tree. More specifically, the sub-network should sustain adversarial omission faults without losing the property. Hence, the sub-network is usually augmented (with edges) from the input network in comparison to its corresponding non-fault-tolerant counterpart. Naturally, an additional goal is to compute a small such sub-network. In contrast, we design a network that is reinforced (or augmented) by additional edges and nodes so that a given routing scheme can be simulated while facing randomized Byzantine faults. As we ask for being able to ``reproduce'' an arbitrary routing scheme (in the sense of a simulation relation), we cannot rely on a sub-network.

The literature also considered random fault models.
In the network reliability problem, the goal is to compute the probability that the (connected) input network becomes disconnected under random independent edge failures. The reliability of a network is the probability that the network remains connected after this random process.
Karger~\cite{karger2001randomized} gave a fully polynomial randomized approximation scheme for the network reliability problem.
Chechik et.~al~\cite{chechik2012sparse} studied a variant of the task, in which the goal is to compute a sparse sub-network that approximates the reliability of the input network.
We, on the other hand, construct a reinforced network that increases the reliability of the input network;
note also that our requirements are much stricter than merely preserving connectivity.

%###
\paragraph*{Self-healing systems}
%###
In the context of self-healing routing (e.g., Casta\~{n}eda et~al.~\cite{Castaneda2016}), researchers have studied a model where an adversary removes nodes in an online fashion, one node in each time step (at most $n$ such steps). In turn, the distributed algorithm adds links and sends at most $O(\Delta)$ additional messages to overcome the inflicted omission fault.
Ideally, the algorithm is ``compact'': each node's storage is limited to $o(n)$ bits.
A nice property of the algorithm in~\cite{Castaneda2016} is that the degrees are increased by at most $3$. For our purposes, an issue is that the diameter is increased by a logarithmic factor of the maximum initial degree, and hence the same holds for the latency of the routing scheme. Instead, we design a network that is ``oblivious'' to faults in the sense that the network is ``ready'' for independent random faults up to a certain probability, without the need to reroute messages or any other reconfiguration. Moreover, our reinforcements tolerate Byzantine faults and work for arbitrary routing schemes. We remark that compact self-healing routing schemes also deal with the update time of the local data structures following the deletion of a node; no such update is required in our approach.

%###
\paragraph*{Robust Peer-to-Peer Systems}
%###
Peer-to-peer systems are often particularly dynamic and the development
of robust algorithms hence crucial.
Kuhn et.~al~\cite{kuhn2010towards} study faults in peer-to-peer systems in which an adversary adds and removes nodes from the network within a short period of time (this process is also called churn). In this setting, the goal is to maintain functionality of the network in spite of this adversarial process. Kuhn et~al.~\cite{kuhn2010towards} considered hypercube and pancake topologies, with a powerful adversary that cannot be ``fooled'' by randomness. However, it is limited to at most $O(\Delta)$ nodes, where $\Delta$ is the (maximum) node degree, which it can add or remove within any constant amount of time. The main idea in~\cite{kuhn2010towards} is to maintain a balanced partition of the nodes, where each part plays the role of a supernode in the network topology. This is done by rebalancing the nodes after several adversarial acts, and increasing the dimensionality of the hypercube in case the parts become too big.

Hypercubes were also of particular interest in this paper. We employ two partitioning techniques to make sure that: (1)~the size of each part is constant and (2)~the number of links in the cut between the parts is at most $\eps\cdot n$, where $n$ is the number of nodes. These partitioning techniques help us dial down the overheads within each part, and avoid a failure of each part due to its small size. However, we note that our motivation for considering these topologies is that they are used as communication topologies, for which we can provide good reinforcements, rather than choosing them to exploit their structure for constructing efficient and/or reliable routing schemes (which is of course one, but not the only reason for them being used in practice).

\section{Conclusion}\label{sec:conc}

In this paper, we proposed simple replication strategies for improving network reliability. Despite being simple and general, both in terms of their application and analysis, our strategies can substantially reduce the required reliability on the component level to maintain network functionality compared the baseline, without losing messages or increasing latencies.
The presented transformations allow us to directly reuse non-fault-tolerant routing schemes as a blackbox,
and hence avoid the need to refactor working solutions.
We consider this property highly useful in general and essential in real-time systems.
%While weaker guarantees may result in more efficient solutions, they also necessitate that other system levels must be able to handle the consequences.

Hence, being prepared for non-benign faults can be simple, affordable, and practical, and therefore enables building larger reliable networks. Interestingly, while our basic schemes may hardly surprise, we are not aware of any work systematically exploring and analyzing this perspective.

We understand our work as a first step and believe that it opens
several interesting avenues for future research.
For example:
\begin{itemize}[(i)]
  \item Which network topologies allow for good partitions as utilized in \S~\ref{sec:eff}? Small constants here result in highly efficient reinforcement schemes, which are key to practical solutions.
  \item Is it possible to guarantee strong simulations at smaller overheads?
  \item Can constructions akin to the one given in \S~\ref{sec:eff} be applied to a larger class of graphs?
\end{itemize}

On the practical side, while
our simulations indicate that our approach
can be significantly more efficient than a naive one-by-one replication strategy
to provision
dependable ISP networks,
it will be interesting to extend these empirical studies and also consider
practical aspects such as the incremental deployment
in specific networks.

\noindent \textbf{Acknowledgments.}
This project has received funding from the European Research Council (ERC) under the European Union's Horizon 2020 research and innovation programme (grant agreement 716562) and from the Vienna Science and Technology Fund (WWTF), under grant number ICT19-045 (project WHATIF). 
This research was supported by the Israel Science Foundation under Grant 867/19.

\bibliographystyle{spmpsci}
\bibliography{robust}

%\newpage

\begin{IEEEbiography}[{\includegraphics[width=1in,height=1.25in,clip,keepaspectratio]{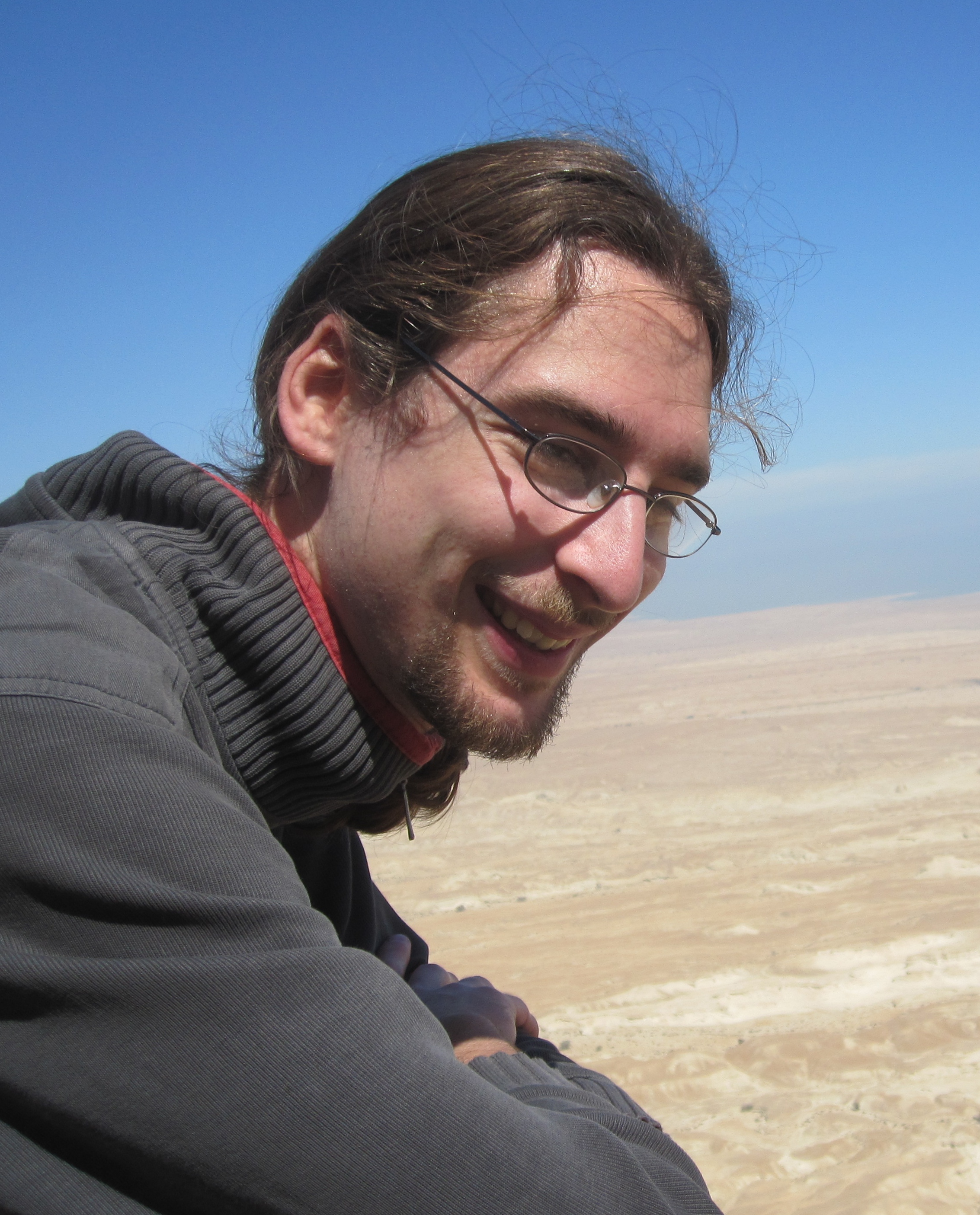}}]{Christoph Lenzen}
received a diploma degree in mathematics from the University of Bonn in 2007 and a
Ph.\,D.\ degree from ETH Zurich in 2011. After postdoc positions at the Hebrew University of Jerusalem,
the Weizmann Institute of Science, and MIT, he became group leader at MPI for Informatics in 2014.
In 2021 he became faculty member at CISPA.
He received the best paper award at PODC 2009, the ETH medal for his dissertation, and in 2017 an ERC starting grant.
\end{IEEEbiography}

\begin{IEEEbiography}[{\includegraphics[width=1in,height=1.25in,clip,keepaspectratio]{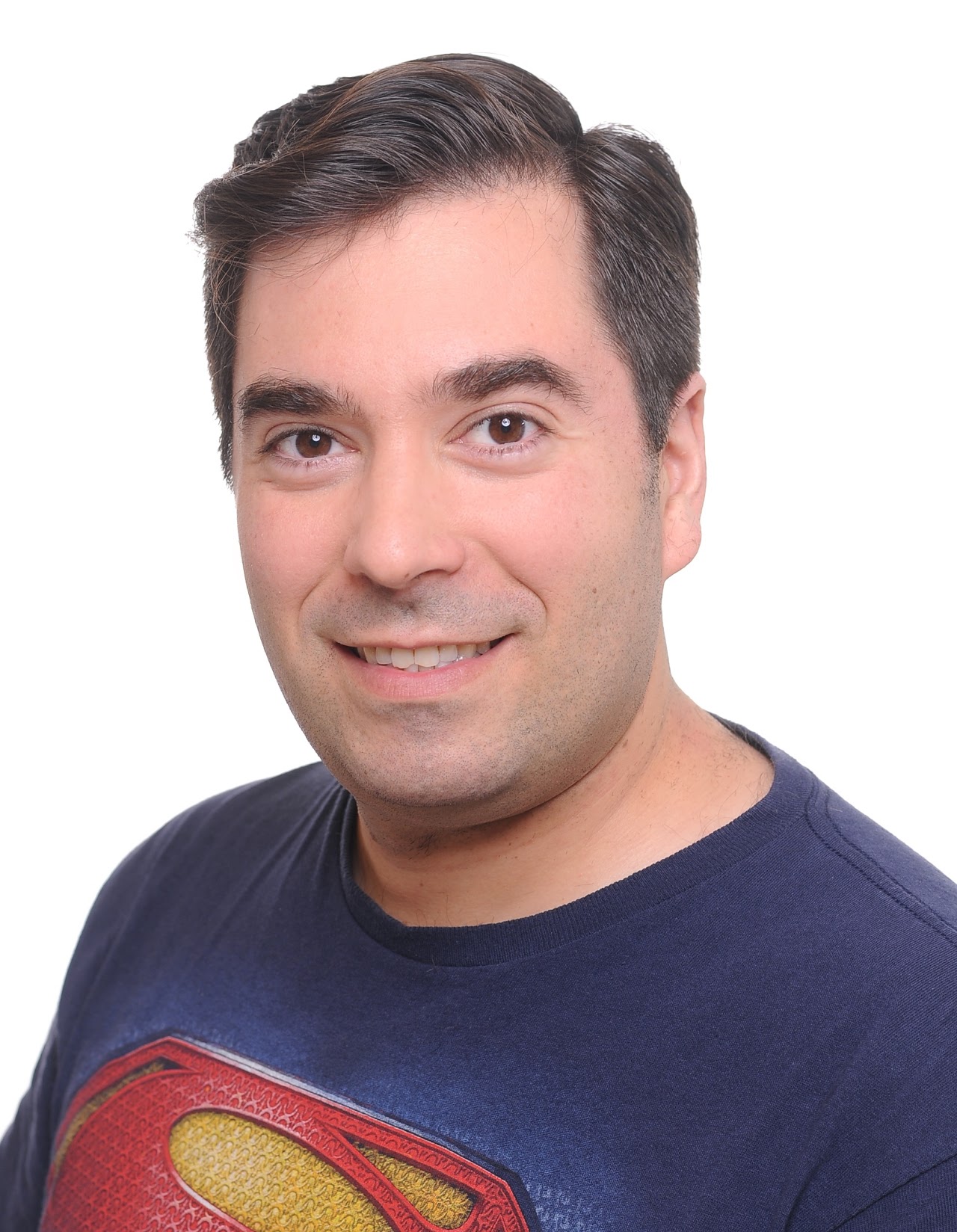}}]{Moti Medina}
is a faculty member at the Engineering Faculty at Bar-Ilan University since 2021. Previously, he was a faculty member at the Ben-Gurion University of the Negev and a post-doc
researcher in MPI for Informatics and in the Algorithms and Complexity group at
LIAFA (Paris 7). He graduated his Ph.\,D., M.\,Sc., and B.\,Sc.\ studies at the
School of Electrical Engineering at Tel-Aviv University, in  2014, 2009, and 2007
respectively. Moti is also a co-author of a  text-book on logic design
``Digital Logic Design: A Rigorous Approach'', Cambridge Univ. Press, Oct.
2012.
\end{IEEEbiography}

\begin{IEEEbiography}[{\includegraphics[width=1in,height=1.25in,clip,keepaspectratio]{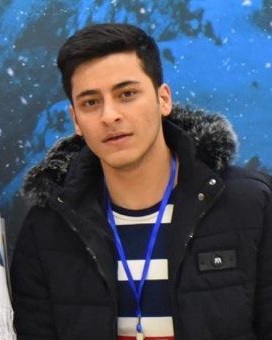}}]{Mehrdad Saberi}
	is an undergraduate student in Computer Engineering at Sharif University of Technology, Tehran, Iran. He achieved a silver medal in International Olympiad in Informatics (2018, Japan) during high school and is currently interested in studying and doing research in Theoretical Computer Science.
\end{IEEEbiography}

\begin{IEEEbiography}[{\includegraphics[width=1in,height=1.25in,clip,keepaspectratio]{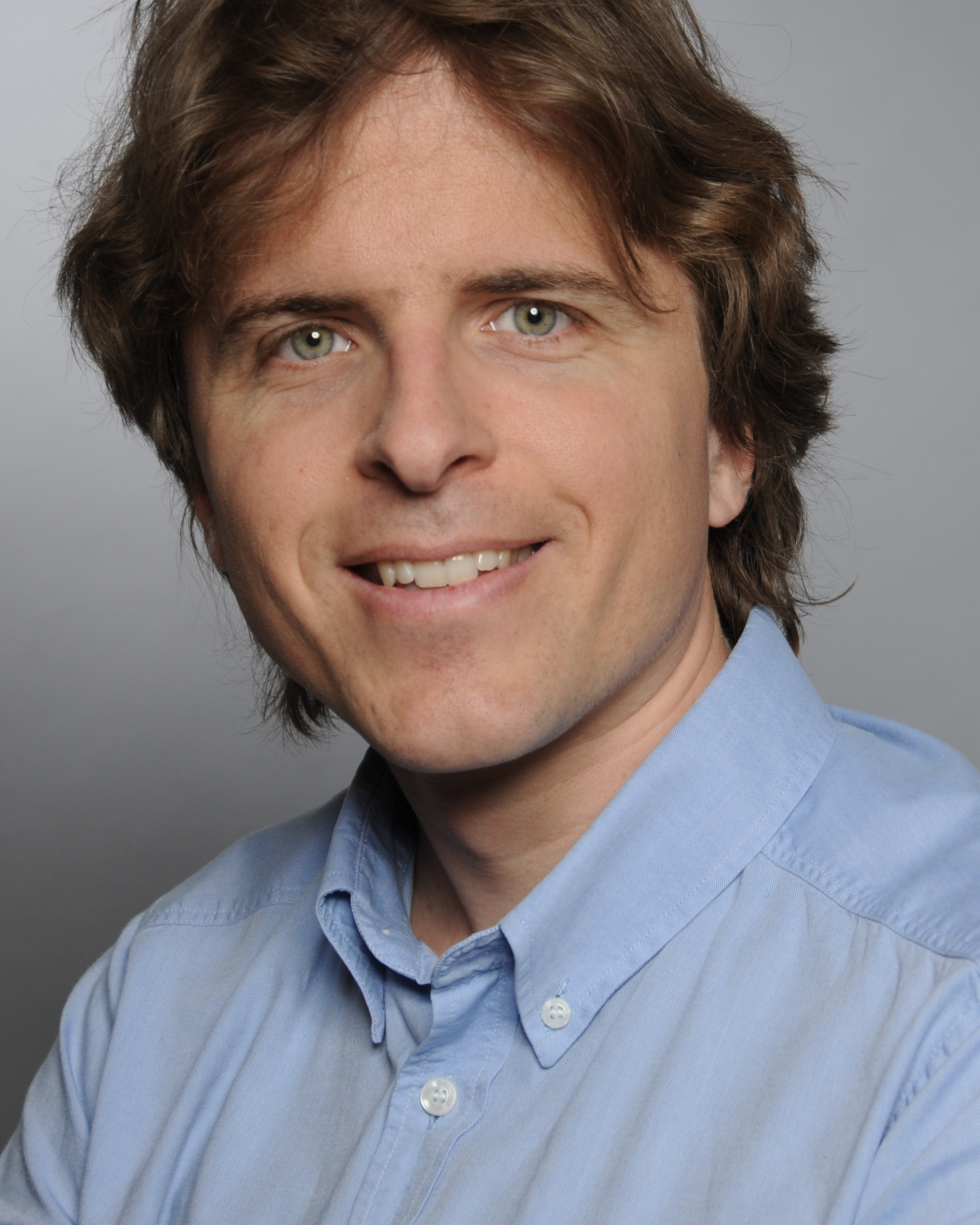}}]{Stefan Schmid}
is a Professor at TU Berlin, Germany. 
He received his MSc (2004) and PhD
(2008) from ETH Zurich, Switzerland. Subsequently, Stefan Schmid
worked as postdoc at TU Munich and the University of Paderborn (2009).
From 2009 to 2015, he was a senior research scientist at the Telekom Innovations Laboratories (T-Labs) in Berlin, Germany, from 2015 to 2018 an Associate
Professor at Aalborg University, Denmark, and from 2018 to 2021 a Professor 
at the University of Vienna, Austria. 
His research interests revolve around algorithmic problems of networked and distributed systems,
currently with a focus on self-adjusting networks
(related to his ERC project AdjustNet) and resilient networks (related to his WWTF project
WhatIf).
\end{IEEEbiography}

\end{document}